\documentclass[pdflatex]{sn-jnl}

\usepackage{amsfonts, amsmath, amsthm, amssymb,mathtools} 
\usepackage{algorithm}
\usepackage{bm}
\usepackage{algpseudocode}
\usepackage{graphicx} 
\usepackage{booktabs} 
\usepackage{natbib}  
\usepackage{subcaption}
\usepackage{placeins}

\theoremstyle{thmstyleone}%
\newtheorem{theorem}{Theorem}
\theoremstyle{thmstyletwo}%
\newtheorem{remark}{Remark}%
\newtheorem{lem}{Lemma}
\theoremstyle{thmstylethree}%
\newtheorem{definition}{Definition}%

\begin{document}

\title[]{Scan Statistics for Nonhomogeneous Poisson Processes with Extreme-Value Calibration and Application to CNV Detection}

\author[1]{\fnm{Asanka R.} \sur{Duwage}}\email{ard490@msstate.edu}
\equalcont{These authors contributed equally to this work.}

\author*[1]{\fnm{Tung-Lung} \sur{Wu}}\email{twu@math.msstate.edu}
\equalcont{These authors contributed equally to this work.}


\affil[1]{\orgdiv{Department of Mathematics and Statistics}, \orgname{Mississippi State University}, \orgaddress{ \city{Mississippi State}, \postcode{39762}, \state{Mississippi}, \country{USA}}}




\abstract{
We develop a scan statistic method for detecting local clusters in a two-sample nonhomogeneous Poisson process (NHPP) framework, motivated by copy number variation (CNV) analysis in next-generation sequencing data. The control sample is used to construct an empirical time transformation, under which the transformed case sample is approximately uniform on [0,1] under the null hypothesis. The scan statistic is defined as the maximum number of transformed points within a moving window.
		
We show that the scan statistic converges to a generalized extreme value (GEV) distribution with an extremal index that captures the dependence induced by overlapping windows. The GEV parameters and extremal index are estimated using maximum likelihood and exceedance clustering methods, providing an asymptotic calibration of the test. A permutation procedure is also developed to provide a nonparametric alternative. Simulation studies show that the permutation calibration maintains empirical Type I error
close to the nominal level across the considered settings, and the GEV calibration is
accurate for smaller windows. Both proposed procedures
show competitive power compared with the continuous testing method under heterogeneous
baseline intensities. An application to sequencing data illustrates the effectiveness of the proposed approach for detecting CNV regions.	
}




\maketitle

\section{Introduction}\label{sec1}

Scan statistics provide a hypothesis-testing framework for detecting localized clusters of events in time or space.
 They are used in many fields, including epidemiology
(\citealt{li_chen_2008}), engineering (\citealt{Neil:2013}), public health (\citealt{kulldorff2009}),
molecular biology (\citealt{Berry2014}), and astronomy (\citealt{Darling:1986}). The main question is
whether an observed local aggregation is unlikely to occur under a baseline model. Classical
scan procedures slide a window across the domain, compute a local count (or related statistic)
inside each window, and reject the null when the maximum over windows is too large relative
to its null distribution.

The scan statistic is an extreme-value functional of many
overlapping-window statistics, which are strongly dependent. Therefore, a complete testing
procedure requires a good approximation of the tail probability of the scan statistic under
the null. Large bodies of work exist for homogeneous Poisson processes and for independent
and identically distributed (i.i.d.) models, including Bernoulli, Binomial, Poisson, Normal,
Exponential, Gamma, and related families (\citealt{textbook1, textbook2}). Further developments
include higher-dimensional scans (\citealt{Glaz_CHEN:1996, Glaz_Gue-2010}), flexible
window shapes (\citealt{Kulldorff_Martin-2006}), variable window sizes
(\citealt{Xiao-2014, Nagarwalla1996, Patil2004}), and a variety of exact and approximate methods
for computing tail probabilities and critical values.

Many modern data sets have \emph{nonhomogeneous} baselines, where the event rate
changes across the domain. This happens often in two-sample point-process problems, where a
``case'' process is compared with a ``control'' process. In such settings, using a
homogeneous null can inflate Type I error, because large local counts may simply reflect a
high baseline intensity rather than a true cluster. The window overlap dependence is still
present, and the baseline heterogeneity adds another layer of difficulty. Several approaches
handle heterogeneity by explicit modeling or approximations (\citealt{ShenZhang2012, WU2017}), but
rigorous null distributions for scan statistics under nonhomogeneous Poisson process (NHPP)
models remain limited, especially in the two-sample framework.

This paper is motivated by copy number variation (CNV) and copy number alteration (CNA)
detection in cancer genomics (\citealt{Shlien2009, Joan2018}). Next-generation sequencing (NGS)
produces millions of short reads that are mapped to genomic positions. These mapped read
locations can be modeled naturally as realizations of point processes. CNVs/CNAs appear as
localized changes in the relative read intensity between tumor and normal samples, often under
strong baseline heterogeneity along the genome. In the change-point framework of
\citet{ShenZhang2012}, the tumor and normal read count processes along a chromosome are modeled
as two NHPPs, and CNV/CNA events correspond to segments where the relative intensity changes.
They also discussed practical limitations of binning or fixed-window methods. Due to read inhomogeneity, the optimal window size may vary across the genome, and fixed binning can lead to imprecise boundaries when breakpoints do not align with bin edges.
 These
observations motivate methods that operate directly on the point-process representation while
still providing valid statistical inference.

We formulate CNV/CNA detection as a \emph{two-sample cluster detection} problem. Let the normal
(control) sample and the tumor (case) sample be independent NHPPs. Under the null, the two intensity functions are proportional, 
\(\lambda_1(t)=\rho\lambda_0(t)\) for some \(\rho>0\); equivalently, conditional on the
total numbers of reads, the case and control event locations have the same distribution.
Under the alternative, there exists an unknown window where the tumor intensity is
locally higher (amplification) or lower (deletion) than this proportional baseline. The scan statistic searches over possible windows and reports the largest local
departure. The key step is to obtain a usable null distribution for the scan statistic when
the baseline is heterogeneous and unknown.

Recent work by \citet{Picard2018} proposed a different perspective through \emph{continuous
testing}. They consider a Poisson process observed on a fixed domain and test a continuum of
local hypotheses indexed by the window location. For a window of fixed width $\eta$, the data
in each window are summarized by the number of events falling inside that window. This leads
to a location-indexed collection of $p$-values (a ``$p$-value process''), where each $p$-value
measures how surprising the observed window count is under the null intensity. Their framework
then applies multiple-testing ideas to control error rates such as FWER or FDR over an
uncountable set of window locations. This approach is well suited for reporting many
significant windows with rigorous error control. In contrast, our focus is on the classical
scan-statistic test based on the maximum over windows and on fast calibration of its critical
value in the two-sample NHPP setting.

Our main contribution is a computationally efficient testing procedure for two-sample NHPP data based on scan statistics that requires only one sample path to estimate reliable critical values. We use the control sample to
construct an empirical time transformation that reduces the two-sample NHPP problem to a
scan on approximately uniform points under the null. We then develop an extreme value theorem
for the resulting scan statistic and use a generalized extreme value (GEV) approximation with
an extremal index to account for window overlap dependence. This yields a practical method
for computing tail probabilities and critical values at high resolution. To complement the
asymptotic calibration and to avoid reliance on parametric baseline specification, we also
use a permutation procedure for nonparametric significance assessment in the two-sample
setting. The proposed methods are evaluated by simulation, including size and power
comparisons with the continuous testing method of \citet{Picard2018}, and are applied to synthetic 
NGS tumor/normal data used in \citet{ShenZhang2012}.

The rest of the paper is organized as follows. In Section~2, we formulate the
two-sample nonhomogeneous Poisson process model and introduce the proposed scan statistic
based on the empirical time transformation. We establish an extreme value theorem for the
scan statistic and describe maximum likelihood estimation of the GEV
distribution parameters and the extremal index. A permutation-based testing procedure is also
presented to provide a fully nonparametric calibration. Section~3 evaluates the performance
of the proposed methods through extensive simulation studies, including empirical size and
power comparisons with the continuous testing method of \citet{Picard2018}. In Section~4, we
apply the proposed methods to synthetic sequencing data for copy number variation
detection and compare the detected regions with existing approaches. Finally, Section~5
concludes with a summary and discussion.

\section{Methodology}

\subsection{Problem formulation}

Copy number variations (CNVs) refer to differences in the number of copies of specific
DNA segments within an individual's genome. These structural variants include genomic
deletions and duplications, often spanning many base pairs, and contribute substantially
to genomic diversity. While CNVs commonly occur in healthy individuals, they also play a
critical role in human disease. In particular, tumor cells frequently exhibit abnormal
copy-number patterns, known as copy number alterations (CNAs), arising from somatic
deletions or duplications that differentiate malignant genomes from their normal
counterparts. CNVs and CNAs correlate strongly with cancer initiation, progression, and
therapeutic response, making their reliable detection a central task in cancer genomics
(\citealt{Shlien2009,Joan2018}).

Similar to \cite{ShenZhang2012}, we model the normal (control) sample and tumor (case) sample as
two independent nonhomogeneous Poisson processes (NHPPs) on $[0,1]$, where the genomic
coordinate is rescaled to $[0,1]$ for convenience. Let $\{X(t),t\in[0,1]\}$ be the control process with
intensity $\lambda_0(t)$, and let $\{Y(t),t\in[0,1]\}$ be the case process with intensity $\lambda_1(t)$.
 A CNV/CNA occurs when, on some
unknown interval, the tumor intensity differs from the normal baseline. We test whether there
exists an interval of length $s$ starting at an unknown location $\tau\in[0,1-s]$ where the
two intensities are different:
\begin{equation}
H_0 : \lambda_1(t)=\rho \lambda_0(t), \quad t\in[0,1],
\quad \text{for some } \rho>0,
\label{eq1}
\end{equation}
\begin{equation}
H_1 :
\begin{cases}
	\lambda_1(t) \neq \rho \lambda_0(t), & t\in[\tau,\tau+s],\\
	\lambda_1(t) = \rho \lambda_0(t), & t\notin[\tau,\tau+s],
\end{cases}
\end{equation}
In this paper, we focus on finding local \emph{excess} tumor events (amplifications). The
same idea can be used for deletions by scanning for unusually small counts.

A classical scan statistic for detecting an excessive number of case events in a window
of size $\omega\in(0,1)$ is
\begin{equation}\label{scan_stat}
	S(\omega)=\sup_{0\le t\le 1-\omega} Y_t(\omega),
	\qquad
	Y_t(\omega)=Y(t+\omega)-Y(t),
\end{equation}
where $Y_t(\omega)$ is the number of case events in the interval $(t,t+\omega]$.

The essential inferential task is to determine whether the local intensity over
$[\tau,\tau+s]$ deviates significantly from the baseline intensity $\lambda_0(t)$.
To calibrate the scan statistic under heterogeneous baselines, we develop an extreme value
approximation for its null distribution and also provide a permutation-based alternative
for nonparametric significance assessment.

\subsection{An extreme value theorem}

Let \(X_1,\ldots,X_{n_0}\) denote the control event locations on \([0,1]\). Conditional on the
total count \(X(1)=n_0\), the event times form an i.i.d.\ sample from
\[
f_X(t)=\frac{\lambda_0(t)}{\Lambda_0(1)}, \quad t\in[0,1],
\qquad
F_X(t)=\frac{\Lambda_0(t)}{\Lambda_0(1)},
\]
where \(\Lambda_0(t)=\int_0^t \lambda_0(s)\,ds\).

Similarly, let \(Y_1,\ldots,Y_{n_1}\) denote the case event locations on \([0,1]\). Conditional
on \(Y(1)=n_1\), the event times form an i.i.d.\ sample from
\[
f_Y(t)=\frac{\lambda_1(t)}{\Lambda_1(1)}, \quad t\in[0,1],
\qquad
F_Y(t)=\frac{\Lambda_1(t)}{\Lambda_1(1)}.
\]
Unless otherwise stated, we work conditionally on $n_0$ and $n_1$. Define the probability integral transform
\begin{equation}\label{Ui_def}
	U_j = F_X(Y_j), \quad j=1,\ldots,n_1.
\end{equation}
Under \(H_0\) in \eqref{eq1}, the proportionality constant \(\rho\) cancels after
conditioning on the total counts, so \(F_Y=F_X\). Hence,
\(U_1,\ldots,U_{n_1}\) are i.i.d. \(\mathrm{Unif}(0,1)\).

Given  \(\omega_{n_1}\in(0,1)\) and a grid spacing \(\Delta_m>0\), define the
grid points by
\begin{equation}
	\label{eq:grid}
	u_i=(i-1)\Delta_m,\qquad i=1,\dots,m_{n_1},
\end{equation}
where
\[
m_{n_1}=\left\lfloor \frac{1-\omega_{n_1}}{\Delta_m}\right\rfloor+1.
\]
Then
\[
0 = u_1 < u_2 < \cdots < u_{m_{n_1}} \le 1-\omega_{n_1},
\]
so every grid window \([u_i,u_i+\omega_{n_1}]\) is contained in \([0,1]\).

Define the grid-based sliding-window count sequence by
\begin{equation}
	\label{gridcount}
	C_i(\omega_{n_1})=\sum_{j=1}^{n_1} \mathbf 1\{U_j\in [u_i,u_i+\omega_{n_1}]\},
	\qquad i=1,\dots,m_{n_1}.
\end{equation}
Since \(\omega_{n_1}\) and \(m_{n_1}\) depend on \(n_1\), we view
\(\{C_i(\omega_{n_1})\}_{i=1}^{m_{n_1}}\) 
as a row-wise  stationary triangular array.
To develop an extreme value theorem for the scan statistic in \eqref{scan_stat}, we show that
\(\{C_i(\omega_{n_1})\}_{i=1}^{m_{n_1}}\)  satisfies the regularity conditions required for extreme value theory.

The following definition and lemma can be found in \cite{textbook5}.

\begin{definition}\label{mix}
	Let \(\{S_1,S_2,\ldots,S_m\}\) be a stationary sequence. For any index sets
	\(1\le i_1<\cdots<i_s<j_1<\cdots<j_t\le m\) such that \(j_1-i_s>l\),
	define
	\begin{multline}\label{mixcond}
		\Big|
		P(S_{i_1}\le u,\ldots,S_{i_s}\le u,\, S_{j_1}\le u,\ldots,S_{j_t}\le u)
		\\
		-
		P(S_{i_1}\le u,\ldots,S_{i_s}\le u)\,
		P(S_{j_1}\le u,\ldots,S_{j_t}\le u)
		\Big|
		\le g(l).
	\end{multline}
	If \(g(l)\to 0\) as \(l\to\infty\), the sequence is said to satisfy the distributional mixing condition.
\end{definition}

\begin{lem}\label{depgevd}
	Let \(\{S_1,S_2,\ldots,S_m\}\) be a stationary sequence with marginal distribution \(F\),
	satisfying \eqref{mixcond}. Let \(\{S_1',S_2',\ldots,S_m'\}\) be an i.i.d.\ sequence with the
	same marginal distribution \(F\). Define
	\[
	M_m=\max_{1\le i\le m} S_i,
	\qquad
	M_m'=\max_{1\le i\le m} S_i'.
	\]
	If the distributional mixing condition in \eqref{mixcond} is satisfied, and there exist
	sequences \(\{a_m>0\}\) and \(\{b_m\}\) such that
	\begin{align}
		P\!\left(\frac{M_m'-b_m}{a_m}\le x\right) &\to G(x),
	\end{align}
	then
	\begin{align}
		P\!\left(\frac{M_m-b_m}{a_m}\le x\right) &\to [G(x)]^{\theta},
	\end{align}
	where \(\theta\in(0,1]\) is the \emph{extremal index}. Depending on \(F\), the limit \(G\) is a
	generalized extreme value (GEV) distribution:
	\begin{equation}\label{gevdd}
		G(x)=
		\begin{cases}
			\exp\!\left(-\left[1+\xi\left(\dfrac{x-\mu}{\sigma}\right)\right]^{-1/\xi}\right),
			& \xi\neq 0,\\[1.2ex]
			\exp\!\left(-e^{-(x-\mu)/\sigma}\right), & \xi=0,
		\end{cases}
	\end{equation}
	for \(x\) satisfying \(1+\xi(x-\mu)/\sigma>0\), where \(\mu\) is the location parameter,
	\(\sigma>0\) is the scale parameter, and \(\xi\) is the shape parameter.
\end{lem}


\begin{lem}\label{stationary}
	Assume that \(H_0\) holds and that the true null distribution function is used
	in the transformation. Conditional on \(n_1\), let
	\(U_1,\ldots,U_{n_1}\) be i.i.d.\ \(\mathrm{Unif}(0,1)\). Given
	\(\omega\in(0,1)\), define
	\[
	C_i(\omega)=\sum_{j=1}^{n_1}
	\mathbf 1\{U_j\in [u_i,u_i+\omega]\},
	\qquad i=1,\ldots,m_{n_1},
	\]
	where \(u_i=(i-1)\Delta_m\) and each interval
	\([u_i,u_i+\omega]\) is contained in \([0,1]\). Then the sequence
	\(\{C_i(\omega)\}_{i=1}^{m_{n_1}}\) is strictly stationary in the finite-index sense:
	for any \(k\ge 1\), any indices \(i_1<\cdots<i_k\), and any integer shift \(r\)
	such that \(1\le i_1+r<\cdots<i_k+r\le m_{n_1}\),
	\[
	(C_{i_1}(\omega),\ldots,C_{i_k}(\omega))
	\overset{d}{=}
	(C_{i_1+r}(\omega),\ldots,C_{i_k+r}(\omega)).
	\]
\end{lem}

\begin{proof}
	For fixed intervals \(A_1,\ldots,A_k\), the vector
	\[
	\left(\sum_{j=1}^{n_1}\mathbf 1\{U_j\in A_1\},\ldots,
	\sum_{j=1}^{n_1}\mathbf 1\{U_j\in A_k\}\right)
	\]
	has a multinomial distribution over the finite partition generated by
	\(A_1,\ldots,A_k\). Hence, its distribution is determined by the
	lengths of all intersections and set differences generated by these intervals.
	
	Now take
	\[
	A_\ell=[u_{i_\ell},u_{i_\ell}+\omega],
	\qquad
	A_\ell^{(r)}=[u_{i_\ell+r},u_{i_\ell+r}+\omega].
	\]
	Since \(u_{i_\ell+r}=u_{i_\ell}+r\Delta_m\), the shifted collection
	\(\{A_\ell^{(r)}\}_{\ell=1}^k\) is obtained from
	\(\{A_\ell\}_{\ell=1}^k\) by a common translation. Because all intervals remain
	inside \([0,1]\), this translation preserves the length of every intersection
	and every cell in the partition generated by the intervals. Therefore, the two
	multinomial cell-probability vectors are identical, and the two count vectors
	have the same distribution.
\end{proof}

 \begin{lem}\label{lem:mixing}
 	Assume
 	\[
 	n_1\omega_{n_1}\to\infty,
 	\qquad
 	n_1\omega_{n_1}^2\to0.
 	\]
 	Then \(\{C_i(\omega_{n_1})\}_{i=1}^{m_{n_1}}\) satisfies the asymptotic distributional mixing
 	condition in \eqref{mixcond}.
 \end{lem}

 \begin{proof}
 	Let
 	\[
 	X=(C_{i_1},\ldots,C_{i_s}),
 	\qquad
 	Y=(C_{j_1},\ldots,C_{j_t}),
 	\]
 	where \(1\le i_1<\cdots<i_s<j_1<\cdots<j_t\le m_{n_1}\). Choose
	\[
l=\left\lceil \frac{\omega_{n_1}}{\Delta_m}\right\rceil+1.
\]
 	If \(j_1-i_s>l\), then
 	\[
 	u_{j_1}-u_{i_s}
 	=
 	(j_1-i_s)\Delta_m
 	>
 	\omega_{n_1},
 	\]
 	so every window in the first block is disjoint from every window in the second block.
 	
 	Let \(N\sim\mathrm{Poisson}(n_1)\), independent of the \(U_r\)'s, and define
 	\[
 	C_i^*(\omega_{n_1})
 	=
 	\sum_{r=1}^{N}
 	\mathbf 1\{U_r\in [u_i,u_i+\omega_{n_1}]\}.
 	\]
 	In the Poissonized model, counts over disjoint intervals are independent. Hence
 	\[
 	P^*(X\le u,Y\le u)
 	=
 	P^*(X\le u)P^*(Y\le u).
 	\]
 	
 	The windows involved in \(X\) and \(Y\) generate at most \(K\le 2(s+t)\) disjoint
 	subintervals. Let their probabilities be \(p_1,\ldots,p_K\). Since each subinterval
 	has length at most \(\omega_{n_1}\),
 	\[
 	\sum_{k=1}^K p_k^2
 	\le
 	K\omega_{n_1}^2.
 	\]
 	Let \((Z_1,\ldots,Z_K)\) be the corresponding multinomial counts under the
 	fixed-\(n_1\) model, and let \((Z_1^*,\ldots,Z_K^*)\) be the corresponding independent
 	Poisson counts under the Poissonized model. By the multinomial-to-Poisson total
 	variation bound (see, e.g. \cite{Le_Cam_1960} and \cite{Deheuvels_1988}),
 	\[
 	d_{\mathrm{TV}}
 	\{(Z_1,\ldots,Z_K),(Z_1^*,\ldots,Z_K^*)\}
 	\le
 	2n_1\sum_{k=1}^K p_k^2
 	\le
 	2K n_1\omega_{n_1}^2.
 	\]
 	Since \(X\), \(Y\), and \((X,Y)\) are functions of these cell counts, the same bound
 	applies to their joint and marginal distributions. Therefore, using the independence
 	of \(X\) and \(Y\) under the Poissonized model,
 	\[
 	\left|
 	P(X\le u,Y\le u)-P(X\le u)P(Y\le u)
 	\right|
 	\le
 	C_{s,t}n_1\omega_{n_1}^2,
 	\]
 	where \(C_{s,t}>0\) depends only on the fixed block sizes \(s\) and \(t\). Since
 	\[
 	n_1\omega_{n_1}^2\to0,
 	\]
 	we have
 	\[
 	g(l)\le C_{s,t}n_1\omega_{n_1}^2\to0.
 	\]
 	Thus the asymptotic distributional mixing condition follows.
 \end{proof}

 Next, the grid-based scan statistic is defined by
 \begin{equation}
 	\label{gridscanstat}
 	S_m(\omega_{n_1})=\max_{1\le i\le m_{n_1}} C_i(\omega_{n_1}),
 \end{equation}
 whereas the classic scan statistic is defined by
 \begin{equation}
 	\label{transscanstat}
 	S(\omega_{n_1})
 	=
 	\sup_{0\le u\le 1-\omega_{n_1}}
 	\sum_{j=1}^{n_1}\mathbf 1\{U_j\in [u,u+\omega_{n_1}]\}.
 \end{equation}
 
 \cite{WU2017} proved that the related discrete grid-based scan statistic converges to the continuous one in distribution. We need a stronger result in the following lemma.

 \begin{lem}\label{lem4}
	If $n_1^2\Delta_m\to0$ and $\Delta_m = o(\omega_{n_1})$, then
	\[
	P\{S_m(\omega_{n_1})=S(\omega_{n_1})\}\to1.
	\]
\end{lem}

\begin{proof}
	By Lemma 3.1 of Fu, Wu, and Lou (2012), with grid spacing \(\Delta_m\),
	\[
	S_m(\omega)\le S(\omega)\le S_m(\omega+2\Delta_m).
	\]
	Applying the same inequality to \(\omega-2\Delta_m\), we obtain
	\[
	S(\omega-2\Delta_m)\le S_m(\omega)\le S(\omega).
	\]
	Hence, with \(\omega=\omega_{n_1}\),
	\[
	0\le S(\omega_{n_1})-S_m(\omega_{n_1})
	\le
	S(\omega_{n_1})-S(\omega_{n_1}-2\Delta_m).
	\]
	Therefore,
	\[
	P\{S_m(\omega_{n_1})\ne S(\omega_{n_1})\}
	\le
	P\{S(\omega_{n_1})\ne S(\omega_{n_1}-2\Delta_m)\}.
	\]
	
	It remains to show that the probability on the right tends to zero. Let
	\[
	U_{(1)}<\cdots<U_{(n_1)}
	\]
	be the order statistics. If
	\[
	S(\omega_{n_1})\ne S(\omega_{n_1}-2\Delta_m),
	\]
	then increasing the window length from \(\omega_{n_1}-2\Delta_m\) to
	\(\omega_{n_1}\) must include at least one additional point in some scan window.
	Equivalently, there must exist two sample points whose spacing lies within
	\(2\Delta_m\) of \(\omega_{n_1}\). Hence,
	\[
	\{S(\omega_{n_1})\ne S(\omega_{n_1}-2\Delta_m)\}
	\subset
	\left\{
	\exists\, 1\le i<j\le n_1:
	\omega_{n_1}-2\Delta_m
	<
	|U_i-U_j|
	\le
	\omega_{n_1}
	\right\}.
	\]
	
	For independent \(U_i,U_j\sim \mathrm{Unif}(0,1)\), the density of
	\(|U_i-U_j|\) is
	\[
	f(d)=2(1-d),\qquad 0<d<1,
	\]
	so \(f(d)\le2\). Therefore,
	\[
	P\left(
	\omega_{n_1}-2\Delta_m
	<
	|U_i-U_j|
	\le
	\omega_{n_1}
	\right)
	\le
	4\Delta_m.
	\]
	By the union bound,
	\[
	\begin{aligned}
		P\{S_m(\omega_{n_1})\ne S(\omega_{n_1})\}
		&\le
		\sum_{1\le i<j\le n_1}
		P\left(
		\omega_{n_1}-2\Delta_m
		<
		|U_i-U_j|
		\le
		\omega_{n_1}
		\right) \\
		&\le
		4{n_1\choose 2}\Delta_m \\
		&\le
		2n_1^2\Delta_m.
	\end{aligned}
	\]
	Since
	\[
	n_1^2\Delta_m\to0,
	\]
	we have
	\[
	P\{S_m(\omega_{n_1})\ne S(\omega_{n_1})\}\to0.
	\]
	Thus
	\[
	P\{S_m(\omega_{n_1})=S(\omega_{n_1})\}\to1.
	\]
	
\end{proof}

Now, we are ready to prove the main theorem.

\begin{theorem}
	\label{thm:gev-grid-scan} 	
	Assume that $H_0$ holds, and let $U_1,\dots,U_{n_1} \stackrel{\text{i.i.d.}}{\sim} \mathrm{Unif}(0,1)$.
	Let the grid points be
	\[
	u_i=(i-1)\Delta_m,\qquad i=1,\dots,m_{n_1}.
	\]
	Assume
	\[
	\omega_{n_1}\to0,\qquad
	n_1\omega_{n_1}\to\infty,\qquad
	n_1\omega_{n_1}^2\to0,
	\]
	\[
	n_1^2\Delta_m\to0,
	\qquad
	\log m_{n_1}=o\left((n_1\omega_{n_1})^{1/3}\right),
	\]
	where
	\[
	m_{n_1}=
	\left\lfloor\frac{1-\omega_{n_1}}{\Delta_m}\right\rfloor+1.
	\]
	Let $C_i(\omega_{n_1})$ and $S_m(\omega_{n_1})$ be defined as in (\ref{gridcount}) and (\ref{gridscanstat}), respectively.
	
	Then there exist normalizing constants $\{a_{n_1}>0\}$ and $\{b_{n_1}\}$ such that, as $n_1\to\infty$,
	\[
	\Pr\!\left(\frac{S_m(\omega_{n_1})-b_{n_1}}{a_{n_1}}\le x\right)
	\longrightarrow
	[G(x)]^\theta,
	\qquad x\in\mathbb R,
	\]
	for some $\theta\in(0,1]$. Here, $G(x)$ is the GEV distribution in \eqref{gevdd} with $\xi=0$, namely the Gumbel distribution.
	
	Moreover,
	\[
	\Pr\!\left(\frac{S(\omega_{n_1})-b_{n_1}}{a_{n_1}}\le x\right)
	\longrightarrow
	[G(x)]^\theta,
	\qquad x\in\mathbb R.
	\]
\end{theorem}

  \begin{proof}
  	Under \(H_0\), the probability integral transform gives
  	\[
  	U_1,\ldots,U_{n_1}\stackrel{i.i.d.}{\sim}\operatorname{Unif}(0,1).
  	\]
  	For each grid point \(u_i=(i-1)\Delta_m\),
  	\[
  	C_i(\omega_{n_1})
  	=
  	\sum_{j=1}^{n_1}
  	\mathbf 1\{U_j\in [u_i,u_i+\omega_{n_1}]\},
  	\qquad i=1,\ldots,m_{n_1}.
  	\]
  	Hence
  	\[
  	C_i(\omega_{n_1})\sim \operatorname{Bin}(n_1,\omega_{n_1}).
  	\]
  	
  	Let
  	\[
  	F_{n_1}(x)
  	=
  	P\{\operatorname{Bin}(n_1,\omega_{n_1})\le x\}.
  	\]
  	Also define
  	\[
  	\mu_{n_1}=n_1\omega_{n_1},
  	\qquad
  	\sigma_{n_1}^2=n_1\omega_{n_1}(1-\omega_{n_1}),
  	\]
  	and
  	\[
  	d_{m_{n_1}}
  	=
  	\sqrt{2\log m_{n_1}}
  	-
  	\frac{\log\log m_{n_1}+\log(4\pi)}
  	{2\sqrt{2\log m_{n_1}}}.
  	\]
  	Set
  	\[
  	b_{n_1}
  	=
  	\mu_{n_1}+\sigma_{n_1}d_{m_{n_1}},
  	\qquad
  	a_{n_1}
  	=
  	\frac{\sigma_{n_1}}{d_{m_{n_1}}}.
  	\]
  	
  	By the Cramér-type moderate deviation approximation for binomial tails,
  	under
  	\[
  	n_1\omega_{n_1}\to\infty,
  	\qquad
  	n_1\omega_{n_1}^2\to0,
  	\qquad
  	\log m_{n_1}=o\{(n_1\omega_{n_1})^{1/3}\},
  	\]
  	the binomial upper tail is asymptotically equivalent to the normal upper tail
  	at the extreme levels \(b_{n_1}+a_{n_1}x\). Therefore,
  	\[
  	\frac{1-F_{n_1}(b_{n_1}+a_{n_1}x)}
  	{1-F_{n_1}(b_{n_1})}
  	\to e^{-x},
  	\qquad x\in\mathbb R.
  	\]
  	Moreover, by the choice of \(b_{n_1}\),
  	\[
  	m_{n_1}\{1-F_{n_1}(b_{n_1})\}\to1.
  	\]
  	Consequently,
  	\[
  	m_{n_1}\{1-F_{n_1}(b_{n_1}+a_{n_1}x)\}
  	\to e^{-x}.
  	\]
  	
  	Let
  	\[
  	\widetilde C_1,\ldots,\widetilde C_{m_{n_1}}
  	\]
  	be iid random variables with distribution \(F_{n_1}\). Then
  	\[
  	\begin{aligned}
  		P\left(
  		\frac{\max_{1\le i\le m_{n_1}}\widetilde C_i-b_{n_1}}
  		{a_{n_1}}
  		\le x
  		\right)
  		&=
  		F_{n_1}(b_{n_1}+a_{n_1}x)^{m_{n_1}}  \\
  		&=
  		\left[
  		1-\{1-F_{n_1}(b_{n_1}+a_{n_1}x)\}
  		\right]^{m_{n_1}} \\
  		&\to
  		\exp(-e^{-x}).
  	\end{aligned}
  	\]
  	Thus the iid comparison maximum has the standard Gumbel limit.
  	
  	By Lemma 2, the sequence
\(\{C_i(\omega_{n_1})\}_{i=1}^{m_{n_1}}\) 
  	is stationary. By Lemma 3, separated blocks are
  	asymptotically independent under the assumption
  	$
  	n_1\omega_{n_1}^2\to0.
  	$
  	Therefore, applying the extreme-value theorem for stationary sequences with
  	extremal index \(\theta\in(0,1]\), we obtain
  	\[
  	P\left(
  	\frac{S_m(\omega_{n_1})-b_{n_1}}{a_{n_1}}
  	\le x
  	\right)
  	\to
  	\{\exp(-e^{-x})\}^{\theta}
  	=
  	\exp(-\theta e^{-x}).
  	\]
  	
  	Finally, since
$
  	n_1^2\Delta_m\to0
$
  	and
  	$
  	n_1\omega_{n_1}\to\infty,
  	$
  	we have
  	$
  	\Delta_m=o(\omega_{n_1}).
  	$
  	Hence Lemma 4 applies and gives
  	\[
  	P\{S_m(\omega_{n_1})=S(\omega_{n_1})\}\to1.
  	\]
  	Therefore,
  	\[
  	\begin{aligned}
  		\left|
  		P\left(
  		\frac{S(\omega_{n_1})-b_{n_1}}{a_{n_1}}\le x
  		\right)
  		-
  		P\left(
  		\frac{S_m(\omega_{n_1})-b_{n_1}}{a_{n_1}}\le x
  		\right)
  		\right| 
  		 \le
  		P\{S(\omega_{n_1})\ne S_m(\omega_{n_1})\}
  		\to0.
  	\end{aligned}
  	\]
  	Thus
  	\[
  	P\left(
  	\frac{S(\omega_{n_1})-b_{n_1}}{a_{n_1}}
  	\le x
  	\right)
  	\to
  	\exp(-\theta e^{-x}).
  	\]
  	This completes the proof.
  \end{proof}

\begin{remark}
	We use a GEV distribution in Theorem~\ref{thm:gev-grid-scan}  because the scan statistic is the maximum of a dependent stationary sequence of overlapping window counts, and extreme-value theory implies that, after normalization, such maxima converge to a GEV law up to an extremal-index correction. In practice, the normalizing constants and the exact finite-sample tail behavior are difficult to obtain explicitly, especially under dependence. Therefore, we estimate the GEV parameters by maximum likelihood, which provides a convenient and effective way to calibrate critical values directly from the observed extreme behavior of the scan statistic.
\end{remark}

\begin{remark}
	The GEV conclusion in Theorem~\ref{thm:gev-grid-scan} also applies to the
	unconditional formulation. In the unconditional case, after the time
	transformation under \(H_0\), the transformed case process is a homogeneous
	Poisson process on \([0,1]\). Therefore, counts over disjoint intervals are
	independent. Hence, when two blocks of grid windows are separated by more than
	\[
	\left\lceil \frac{\omega_{n_1}}{\Delta_m}\right\rceil+1,
	\]
	the corresponding block-count vectors are independent in the unconditional
	Poisson model. Thus the distributional mixing condition holds more directly
	than in the conditional fixed-\(n_1\) case, where the multinomial constraint
	induces weak dependence among disjoint intervals and requires the
	Poissonization argument in Lemma~3.
	
	Consequently, the conditional case treated in Theorem~\ref{thm:gev-grid-scan}
	is the more difficult case. Since the unconditional case has independent
	increments over disjoint intervals, the same GEV approximation, with the
	extremal-index correction for overlapping windows, remains valid under analogous shrinking-window and normalization conditions.
\end{remark}

\begin{remark}[Practical choice of the window size and grid size]
	The assumptions in Theorem~\ref{thm:gev-grid-scan} provide useful guidance for
	choosing the window size in applications. Writing
	\[
	\mu_{n_1}=n_1\omega_{n_1},
	\]
	the conditions
	\[
	n_1\omega_{n_1}\to\infty,
	\qquad
	n_1\omega_{n_1}^2\to0
	\]
	are equivalent to
	\[
	\mu_{n_1}\to\infty,
	\qquad
	\frac{\mu_{n_1}^2}{n_1}\to0.
	\]
	Thus, the expected null count in a scanning window should be moderately large, but
	still small relative to \(\sqrt{n_1}\). A natural default choice is
	\[
	\omega_{n_1}=c n_1^{-2/3}.
	\]
	
	In finite samples, \(c\) may be chosen so that \(n_1\omega_{n_1}\) is reasonably large
	while \(n_1\omega_{n_1}^2\) remains small. For example, with \(n_1=5000\),
	\(\omega=0.005\) gives
	\[
	n_1\omega=25
	\qquad\text{and}\qquad
	n_1\omega^2=0.125,
	\]
	which is consistent with the shrinking-window regime. In contrast, \(\omega=0.05\)
	gives
	\[
	n_1\omega=250
	\qquad\text{and}\qquad
	n_1\omega^2=12.5,
	\]
	which is farther from the asymptotic regime and may require more reliance on
	permutation calibration.

\end{remark}

Since $F_X$ is unknown, we estimate it from the control sample using the empirical
CDF. Denote the estimator by $\widehat F_X$ 
  and define
\begin{equation}\label{V}
V_j=\widehat F_X(Y_j),\qquad j=1,\dots,n_1.
\end{equation}
Theorem 1 is stated for the oracle transformation \(F_X\). In practice, \(F_X\)
is replaced by \(\widehat F_X\). The resulting empirical-transform procedure is
validated numerically in Sections 3 and 4, while the permutation calibration
provides finite-sample validity under label exchangeability.

 The corresponding grid-based sliding-window count sequence and scan statistic are then defined by
\[
\widehat C_i(\omega_{n_1})
=\sum_{j=1}^{n_1} \mathbf 1\{V_j\in [u_i,u_i+\omega_{n_1}]\},
\qquad i=1,\dots,m_{n_1},
\]
and
\begin{equation}\label{Sn_def}
	\widehat S_m(\omega_{n_1})
	=
	\max_{1\le i\le m_{n_1}} \widehat C_i(\omega_{n_1}).
\end{equation}

\subsection{Maximum likelihood estimation of parameters $\mu$, $\sigma$, $\xi$, and $\theta$}

\citet{Prescott:1980, Prescott:1983} and \citet{Hosking:1985} studied maximum likelihood estimation for the GEV distribution and related extreme value models. The following approach was discussed in \citet{textbook5} and \citet{textbook8}, and applied to scan statistics in  \citet{Wu:2020}.

Consider the sequence $\mathbb{S}=\{S_1,\ldots,S_m\}$, which represents the sliding-window counts computed from moving windows. This sequence is dependent and stationary. To remove dependence while preserving the marginal distribution, we construct an independent sample
\[
\mathbb{S}^*=\{S_1^*,\ldots,S_m^*\}
\]
by sampling with replacement from $\mathbb{S}$. This bootstrap sample has the same marginal distribution $F(s)$ but does not retain the dependence structure.

Let $u$ denote a high threshold, typically chosen as the 0.95 or 0.99 quantile of the sample. Let $n_u$ denote the number of exceedances above $u$, and denote the exceedance observations by
\[
\{S_{i_1}^*,S_{i_2}^*,\ldots,S_{i_{n_u}}^*\}.
\]

We then form the sequence of points
\begin{equation*}
	N_m=\left\{
	\left(\frac{i_k}{m}, \frac{S_{i_k}^*-b_m}{a_m}\right),\;
	k=1,\ldots,n_u
	\right\},
\end{equation*}
where $a_m>0$ and $b_m$ are normalizing constants. This sequence converges to a marked Poisson process \citep{textbook5,textbook7}. The intensity function of the limiting process on the set
\[
A=[s,t]\times(x,\infty)
\]
is given by
\begin{equation*}
	\label{MLEL}
	\lambda'(x)
	=
	\begin{cases}
		\dfrac{1}{\sigma}
		\left[
		1+\xi\left(\dfrac{x-\mu}{\sigma}\right)
		\right]^{-1/\xi -1},
		& \xi\neq0,\\[2ex]
		\dfrac{1}{\sigma}
		\exp\left(-\dfrac{x-\mu}{\sigma}\right),
		& \xi=0.
	\end{cases}
\end{equation*}
 
Based on the exceedances, the likelihood function of the marked Poisson process is

\begin{equation}
L(\mu,\sigma,\xi)\propto
\begin{cases}
	\exp\!\left(
	-
	\left[
	1+\xi
	\left(
	\dfrac{u-\mu}{\sigma}
	\right)
	\right]^{-1/\xi}
	\right)
	\displaystyle\prod_{k=1}^{n_u}
	\dfrac{1}{\sigma}
	\left[
	1+\xi
	\left(
	\dfrac{S_{i_k}^*-\mu}{\sigma}
	\right)
	\right]^{-1/\xi -1},
	& \xi\neq 0,
	\\[8pt]
	\exp\!\left(
	-
	\exp\!\left(
	-\dfrac{u-\mu}{\sigma}
	\right)
	\right)
	\displaystyle\prod_{k=1}^{n_u}
	\dfrac{1}{\sigma}
	\exp\!\left(
	-\dfrac{S_{i_k}^*-\mu}{\sigma}
	\right),
	& \xi=0.
\end{cases}\label{likeli}
\end{equation}

The likelihood function does not have a closed-form solution, so the estimators
$\hat\mu$, $\hat\sigma$, and $\hat\xi$ are obtained using numerical optimization methods such as maximum likelihood estimation.

\subsubsection*{Estimation of the extremal index $\theta$}

The extremal index $\theta$ measures the degree of clustering of exceedances in dependent sequences. Methods for estimating $\theta$ were studied in
\cite{Leadbetter:1989},
\cite{Hsing:1991},
\cite{Smith:1994},
and \cite{Suveges:2007}.
Further discussion is provided in \cite{textbook5}.

Let $n_u$ denote the number of exceedances above the threshold $u$, and let $n_c$ denote the number of clusters formed by these exceedances. A simple estimator of the extremal index is given by
\begin{equation}
	\label{esttheta}
	\hat{\theta}
	=
	\frac{n_c}{n_u}.
\end{equation}
An exceedance cluster is defined as a maximal consecutive run of indices \(i\) for which \(S_i>u\).

Algorithm~\ref{alg:cap} summarizes the implementation of the testing procedure. 
\begin{algorithm}[hb!]
	\footnotesize
	\caption{GEV-based calibration procedure}
	\label{alg:cap}
	\begin{algorithmic}[1]
		
		\State \textbf{Initialization:}
		\( \text{rep} \), \( \omega \), \( \alpha \), threshold \( u \),
		control sample \(X=(X_1,\ldots,X_{n_0})\), and case sample
		\(Y=(Y_1,\ldots,Y_{n_1})\).
		
		\State Compute the transformed case observations
		$
		V_i, i=1,\ldots,n_1,
		$
		in (\ref{V}).
		
		\State Evaluate the ordered sliding-window count sequence
		$
		\mathbb{S}=\{S_1,\ldots,S_m\}
		$
		and compute the observed scan statistic
		$
		\widehat S_m(\omega_{n_1})=\max_{1\leq i\leq m} S_i .
		$
		
		\State Find the exceedances of \(\mathbb{S}\) above the threshold \(u\):
		$
		\{S_i:S_i>u,\ i=1,\ldots,m\}.
		$
		
		\State Let \(n_u\) be the number of exceedances in \(\mathbb{S}\), and let
		\(n_c\) be the number of exceedance clusters. Estimate the extremal index by
		$
		\hat{\theta}
		=
		\frac{n_c}{n_u} .
		$
		
		\For{\texttt{\(j \leftarrow 1\) \textbf{ to } \(\text{rep}\)}}
		
		\State Generate an i.i.d. bootstrap sample
		$
		\mathbb{S}^*=\{S_1^*,\ldots,S_m^*\}
		$
		by sampling with replacement from \(\mathbb{S}\).
		
		\State Find the exceedances of \(\mathbb{S}^*\) above the threshold \(u\):
		\[
		\{S_i^*:S_i^*>u,\ i=1,\ldots,m\}.
		\]
		
		\State Estimate
		$
		\hat{\xi}^{(j)}, \hat{\mu}^{(j)}, \text{ and } \hat{\sigma}^{(j)}
		$
		using the likelihood function in (\ref{likeli}).
		
		\State Compute the \((1-\alpha)\)-quantile
		$
		\text{cri}_\alpha^{(j)}
		$
		of
		$
		G^{\hat{\theta}}
		\left(
		x\mid
		\hat{\xi}^{(j)},\hat{\mu}^{(j)},\hat{\sigma}^{(j)}
		\right)
		$
		in (\ref{gevdd}).
		
		\EndFor
		
		\State Compute the median critical value
		\[
		\overline{\text{cri}}_\alpha
		=
		\operatorname{median}
		\left\{
		\text{cri}_\alpha^{(1)},\ldots,\text{cri}_\alpha^{(\text{rep})}
		\right\}.
		\]
		
		\State Reject \(H_0\) if
		\[
		\widehat S_m(\omega_{n_1})>\overline{\text{cri}}_\alpha .
		\]
		
	\end{algorithmic}
\end{algorithm}

 \subsection{Permutation test}
 
 Let $X=\{X_1,\ldots,X_{n_0}\}$ be the control sample and $Y=\{Y_1,\ldots,Y_{n_1}\}$ be the
 case sample. Compute the transformed values
 \begin{equation}\label{Vi_perm}
 	V_i=\widehat F_X(Y_i)=\frac{1}{n_0}\sum_{j=1}^{n_0}\mathbb I\{X_j\le Y_i\},
 	\quad i=1,\ldots,n_1.
 \end{equation}
 Let $V_{(1)}\le \cdots \le V_{(n_1)}$ denote the order statistics of $\{V_i\}$.
 For a given window length $\omega\in(0,1)$, define
 \[
\widehat {S}_i(\omega)=\sum_{j=1}^{n_1}\mathbb I\{V_{(i)}\le V_{(j)}<V_{(i)}+\omega\},
 \quad i=1,\ldots,n_1,
 \]
 and the scan statistic
 \begin{equation}\label{Sn_perm}
 	\widehat {S}(\omega)=\max_{1\le i\le n_1} \widehat {S}_i(\omega).
 \end{equation}
 
 To assess significance, form the permutation null by pooling $Z=X\cup Y$ and repeatedly
 randomly relabeling the pooled observations into a ``control'' subset of size $n_0$ and a
 ``case'' subset of size $n_1$. For each relabeling, recompute $\widehat F_X$ from the
 permuted control subset, transform the permuted case subset as in \eqref{Vi_perm}, and
 compute the corresponding scan statistic $\widehat S^{\mathrm{perm}}(\omega)$ defined in
 \eqref{Sn_perm}. Let $\widehat S^{\mathrm{obs}}(\omega)$ denote the observed statistic.
 
 The permutation $p$-value is computed by 
 \[
 \hat p=
 \frac{1+\#\{\widehat S^{\mathrm{perm}}(\omega)\ge \widehat S^{\mathrm{obs}}(\omega)\}}
 {1+B},
 \]
 where $B$ is the number of random permutations.
 
 Under $H_0$, the joint distribution of the pooled sample is invariant to relabeling
 (exchangeability of labels). Since the permutation procedure recomputes the entire
 transformation and statistic after each relabeling, the resulting permutation distribution
 provides a valid finite-sample calibration of $\widehat S(\omega)$.

\FloatBarrier 
\section{Numerical Results}

To evaluate the performance of our proposed methods, we consider the following baseline
intensity functions:
\begin{itemize}
	\item Linear: $\lambda_0(t)=a+bt$,
	\item Exponential: $\lambda_0(t)=a e^{bt}$,
\end{itemize}
where $a>0$ and $b\ge 0$.

In selecting the parameters, we examined the overall behavior of the intensities. The linear intensity function increases at a constant rate, with $b$ directly controlling the growth rate: for $b=0.05$ the growth is slow, for $b=0.5$ the growth is moderate, and for $b=2$ and $b=5$ the growth is rapid. In contrast, the exponential intensity
function exhibits accelerated growth as $t$ increases: for $b=0.05$ the increase is slow,
for $b=0.5$ the growth is moderate, and for $b=2$ and $b=5$ the intensity increases
rapidly.

In the simulations, we use \(n_0=n_1=5000\), with event locations rescaled to
\([0,1]\). 
 Control event locations are generated from the density
\(f_0(t)=\lambda_0(t)/\Lambda_0(1)\), and, under the null hypothesis, case
event locations are generated from the same density.
We use 1000 Monte Carlo replications and 20 bootstrap repetitions for
GEV calibration. The estimator \(\widehat F_X\) is the empirical CDF of the control
sample throughout the numerical study. The permutation calibration uses \(B=\text{1000}\)
random relabelings in each Monte Carlo replication. We considered several
combinations of \((a,b)\), including \((0.05,0.5)\), \((0.5,0.05)\), \((1,2)\), and
\((1,5)\). We used scanning window sizes
\(\omega\in\{0.005,0.01,0.05\}\). The threshold level for GEV fitting was set to
\(p=0.99\), and the nominal significance level was fixed at \(\alpha=0.05\).

\citet{Picard2018} proposed a continuous testing framework that constructs a $p$-value
process based on sliding windows for testing Poisson process intensities. We denote this
approach by CTS. For comparison, CTS was applied using the same transformed scale and the same
window sizes.
 
Figure~\ref{f-1} presents the empirical sizes of the proposed GEV-calibrated scan test (GEV-ST), the permutation test (Permutation), and the competing CTS method under the linear and exponential NHPP models.  Overall, the permutation-based calibration remains close
to the nominal level across the considered settings, while GEV-ST yields
comparable calibration but shows a slightly higher Type I error rate for larger window sizes.
This is consistent with the asymptotic theory. For example, when \(n_1=5000\) and
\(\omega=0.05\), we have
$
n_1\omega=250
\text{ and }
n_1\omega^2=12.5.
$
Although \(n_1\omega\) is large, \(n_1\omega^2\) is not small, so this finite-sample setting is
farther from the shrinking-window regime assumed in Theorem~\ref{thm:gev-grid-scan}.
Larger windows also increase overlap among neighboring scan windows, which can strengthen
local dependence. CTS can be more sensitive to the choice of window size and to baseline
heterogeneity.

\begin{figure}[h!]
	\centering
	\begin{subfigure}[t]{0.48\linewidth}
		\centering
		\includegraphics[width=\linewidth]{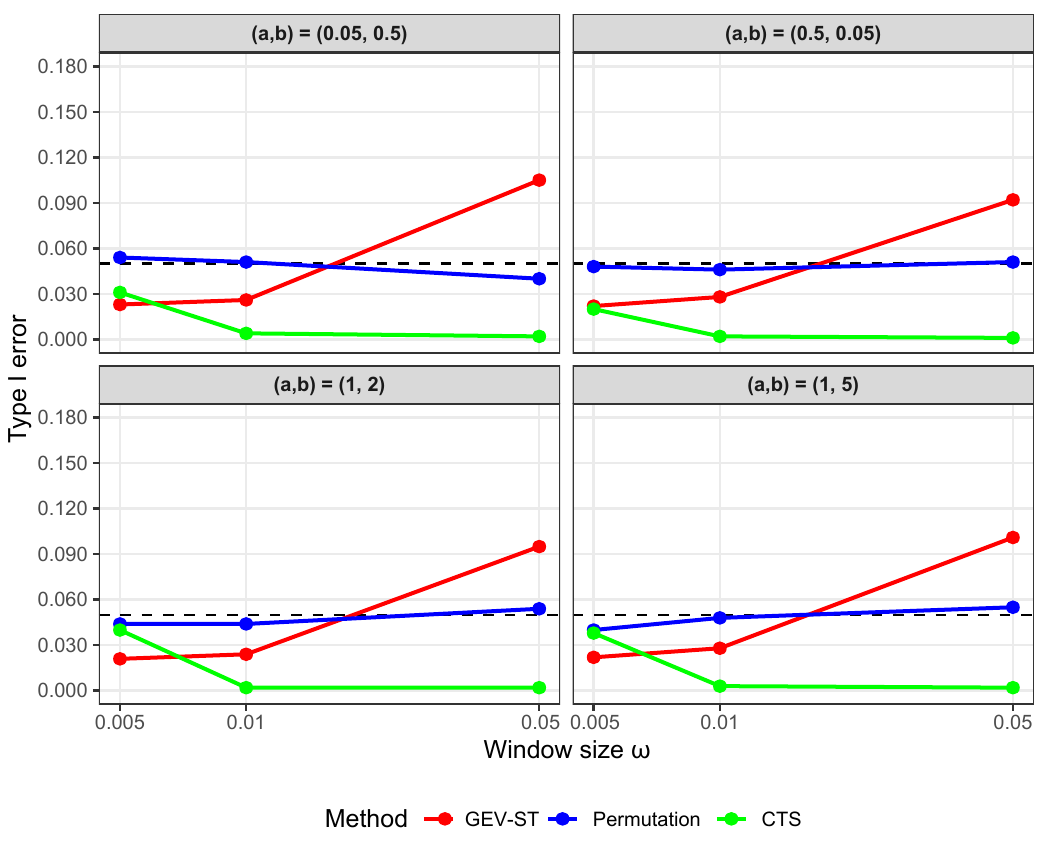}
		\caption{Linear baseline intensity.}
		\label{fig1a}
	\end{subfigure}
	\hfill
	\begin{subfigure}[t]{0.49\linewidth}
		\centering
		\includegraphics[width=\linewidth]{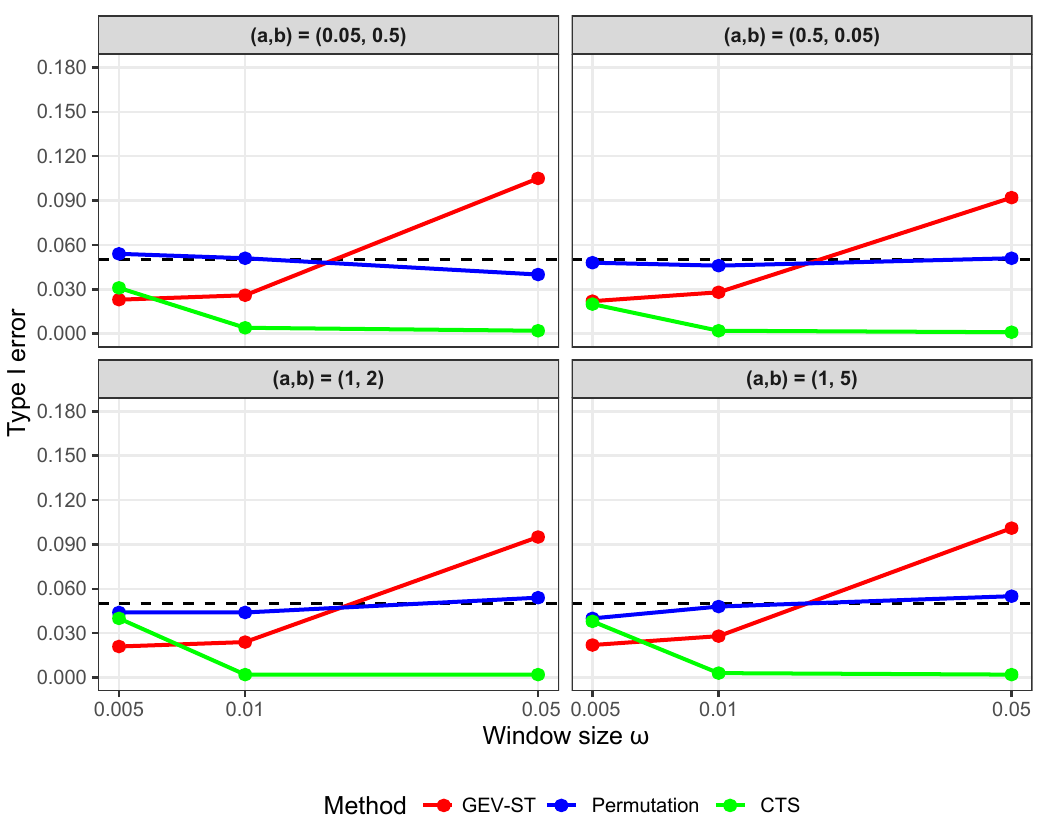}
		\caption{Exponential baseline intensity.}
		\label{fig1b}
	\end{subfigure}
	\caption{Empirical Type I error rates for the GEV-ST, permutation, and CTS methods under linear and exponential NHPP null models.}
	\label{f-1}
\end{figure}

\subsection{NHPP with Embedded Cluster}

Next, we study the power of the proposed methods under alternatives where a true cluster
is present. We consider both linear and exponential baseline intensity functions, the same
settings used in the empirical size study. To create a cluster, we increase the case-process
intensity over a short interval while keeping the rest of the intensity unchanged.

Specifically, for a given baseline intensity function $\lambda_0(t)$, we generate the
case-process intensity as
\begin{equation}
	\label{eq:lambda1_general}
	\lambda_1(t)=\lambda_0(t)+c\,\mathbb{I}\{t\in[\tau,\tau+s]\},
	\qquad 0\le t\le 1,
\end{equation}
where $c>0$ is a constant that controls the strength of the cluster, $\tau$ is the starting
location, and $s$ is the cluster length. The indicator function means that the intensity is
increased by $c$ inside the interval $[\tau,\tau+s]$ and stays equal to the baseline outside
this interval. A larger value of $c$ produces a stronger signal and makes the cluster easier
to detect. In all power simulations, the cluster is embedded in the middle of the sequence.

This model reflects the CNV/CNA setting. The baseline $\lambda_0(t)$ represents the normal
sample, and the increased intensity in $\lambda_1(t)$ represents a copy-number amplification
in the tumor sample. As a result, tumor events occur more frequently in the cluster region.

Let
\[
\Lambda_0(t)=\int_0^t \lambda_0(u)\,du
\]
be the baseline cumulative intensity. Under the alternative model above, the total cumulative
intensity of the case process is $\Lambda_0(1)+cs$. Conditional on the total number of case
events $N=n$, the unordered event times follow the distribution function
\begin{equation}
	\label{eq:F1_piecewise_general}
	F_1(t)=
	\begin{cases}
		\dfrac{\Lambda_0(t)}{\Lambda_0(1)+cs}, & 0\le t<\tau,\\[1.0ex]
		\dfrac{\Lambda_0(t)+c\,(t-\tau)}{\Lambda_0(1)+cs}, & \tau\le t\le \tau+s,\\[1.0ex]
		\dfrac{\Lambda_0(t)+cs}{\Lambda_0(1)+cs}, & \tau+s<t\le 1.
	\end{cases}
\end{equation}

This expression shows that the probability of observing events increases faster inside the
cluster interval. As a result, more case events fall into that region compared to the baseline,
which leads to a larger scan statistic.

In the simulations, we choose the window size $\omega$ as the cluster length $s$ and vary the signal strength $c$ over a range of values to study the performance of each test. For each setting, we generate independent control and case samples from the
corresponding NHPP models, apply the GEV-ST, permutation, and CTS methods,
and estimate power as the proportion of rejections over repeated simulations. Figures~\ref{f-2}--\ref{f-9} show the empirical power of the GEV-ST, permutation, and CTS
methods under both linear and exponential baseline intensity functions for different
parameter values and window sizes. In all figures, the power increases as the signal
strength $c$ increases. This is expected, since larger values of $c$ produce more case
events inside the cluster, making the cluster easier to detect.

Figures~\ref{f-2}--\ref{f-5} correspond to linear baseline intensities with parameter
settings $(a,b)=(0.05,0.5)$, $(0.5,0.05)$, $(1,2)$, and $(1,5)$. In Figure~\ref{f-2},
both GEV-ST and permutation achieve high power even when $c$ is small. In Figure~\ref{f-3}, where
the baseline intensity increases slowly, the overall power is lower, but GEV-ST and
permutation still perform better than CTS. In Figures~\ref{f-4} and \ref{f-5}, which
correspond to faster-increasing baseline intensities, the power of all methods improves
as $c$ increases, but GEV-ST and permutation consistently reach high power faster than CTS.
This shows that the proposed methods are more effective at detecting weaker clusters under
different linear baseline patterns.

Figures~\ref{f-6}--\ref{f-9} present the empirical power of the GEV-ST, permutation, and CTS methods under exponential baseline intensities with parameter settings $(a,b)=(0.05,0.5)$, $(0.5,0.05)$, $(1,2)$, and $(1,5)$, respectively. In all cases, the power increases as the signal strength $c$ increases, and the window size has a substantial effect on performance.
In Figure~\ref{f-6}, GEV-ST and permutation show substantially higher power than CTS for small and moderate window sizes ($\omega=0.005$ and $\omega=0.01$), reaching high power at relatively small signal strengths. For the largest window size ($\omega=0.05$), all methods eventually achieve high power.
In Figure~\ref{f-7}, the overall power is low for small and moderate window sizes, indicating that detecting narrow clusters is difficult under this slowly increasing exponential baseline. However, for the larger window size ($\omega=0.05$), GEV-ST shows a clear advantage, with power increasing sharply at higher signal strengths, while permutation and CTS remain lower.
In Figure~\ref{f-8}, all methods have very low power for small and moderate window sizes. In contrast, for the largest window size ($\omega=0.05$), CTS shows a much faster increase in power and clearly outperforms GEV-ST and permutation, indicating improved sensitivity for detecting broader clusters under more rapidly increasing baselines.
In Figure~\ref{f-9}, a similar pattern is observed under the most rapidly increasing baseline. The power remains low for small and moderate window sizes, while for $\omega=0.05$, all methods achieve high power. 
Overall, GEV-ST and permutation perform better under moderate exponential baselines and for smaller window sizes, whereas CTS becomes more effective when the baseline increases rapidly and the cluster size is large.

Across all figures,  GEV-ST and permutation perform well under slowly or moderately increasing exponential baselines and for smaller window sizes, whereas CTS becomes more competitive and can achieve higher power when the baseline intensity increases rapidly and the window size is large.

\begin{figure}[h!]
	\centering
	\begin{minipage}[t]{0.48\linewidth}
		\centering
		\includegraphics[width=\linewidth]{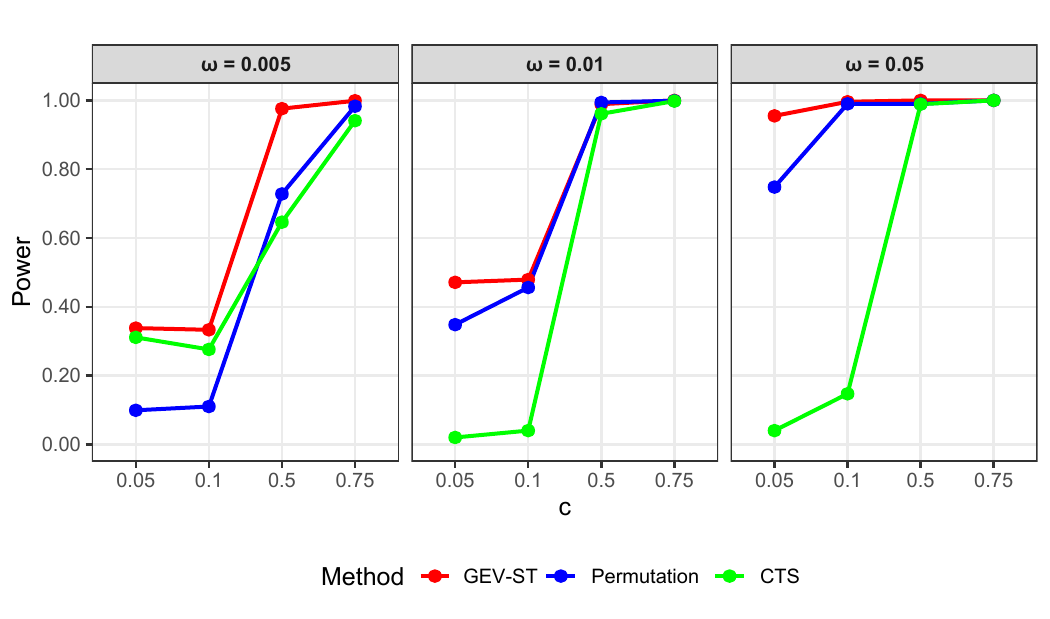}
		\caption{Empirical power under linear baseline intensity with $a=0.05$, $b=0.5$ for different window sizes and signal strengths $c$.}
		\label{f-2}
	\end{minipage}
	\hfill
	\begin{minipage}[t]{0.48\linewidth}
		\centering
		\includegraphics[width=\linewidth]{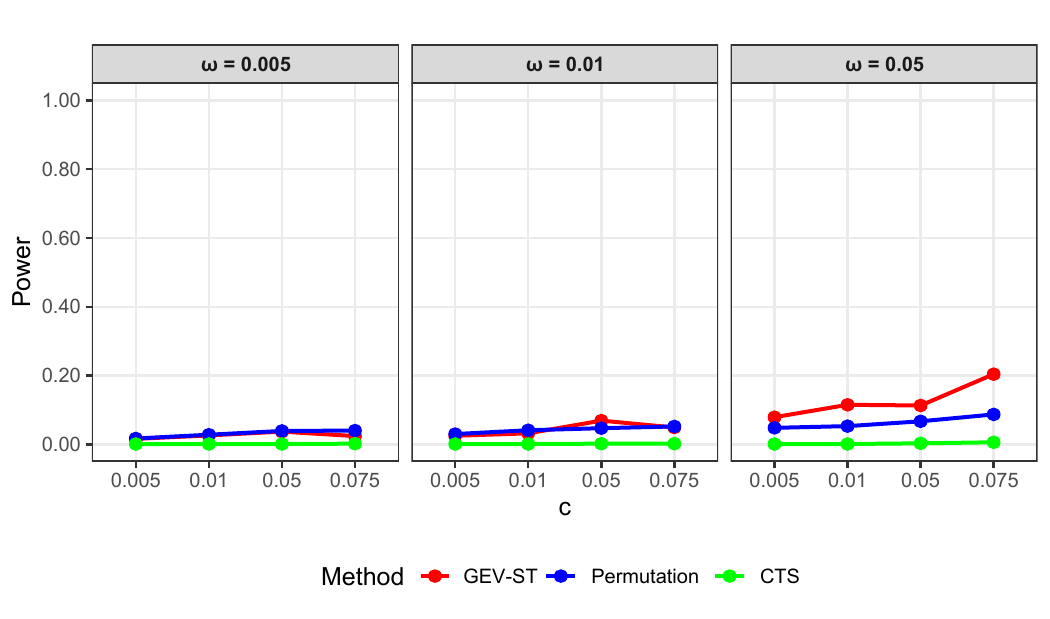}
		\caption{Empirical power under linear baseline intensity with $a=0.5$, $b=0.05$ for different window sizes and signal strengths $c$.}
		\label{f-3}
	\end{minipage}
\end{figure}

\begin{figure}[h!]
	\centering
	\begin{minipage}[t]{0.48\linewidth}
		\centering
		\includegraphics[width=\linewidth]{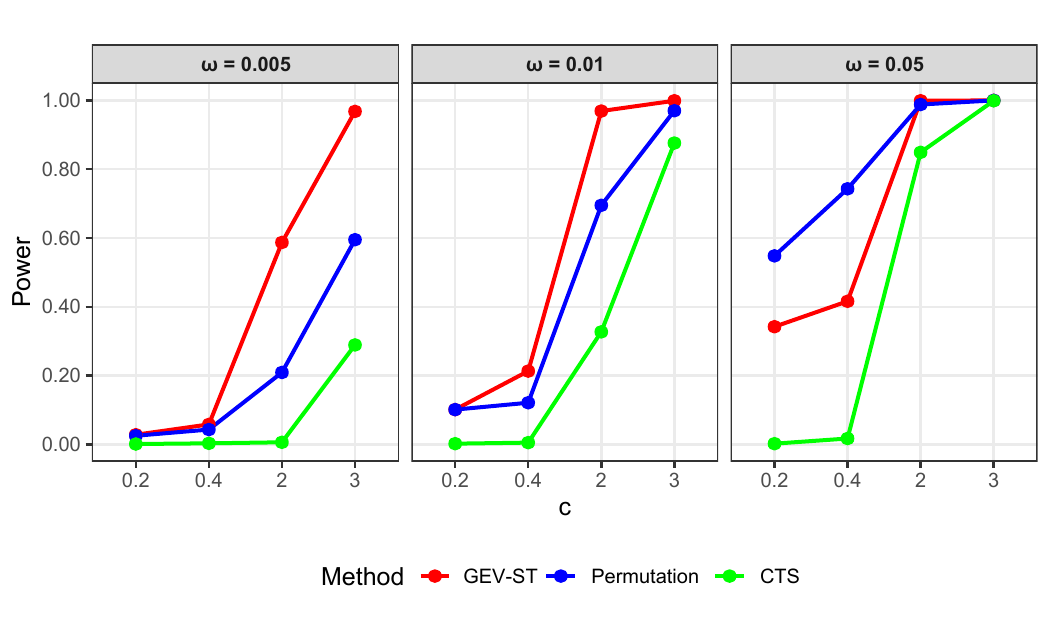}
		\caption{Empirical power under linear baseline intensity with $a=1$, $b=2$ for different window sizes and signal strengths $c$.}
		\label{f-4}
	\end{minipage}
	\hfill
	\begin{minipage}[t]{0.48\linewidth}
		\centering
		\includegraphics[width=\linewidth]{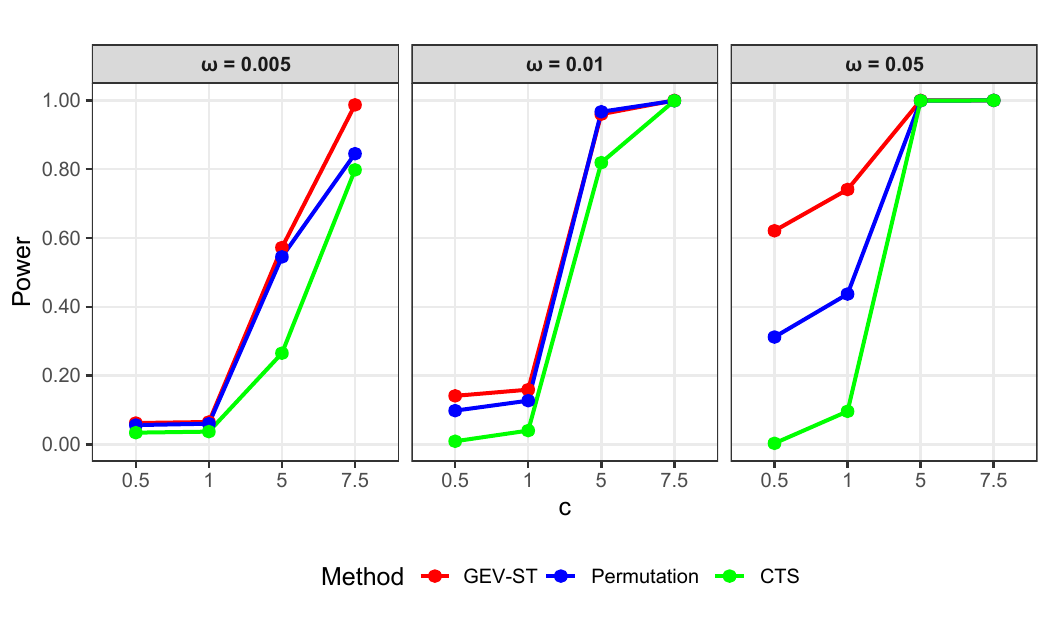}
		\caption{Empirical power under linear baseline intensity with $a=1$, $b=5$ for different window sizes and signal strengths $c$.}
		\label{f-5}
	\end{minipage}
\end{figure}

\begin{figure}[h!]
	\centering
	\begin{minipage}[t]{0.48\linewidth}
		\centering
		\includegraphics[width=\linewidth]{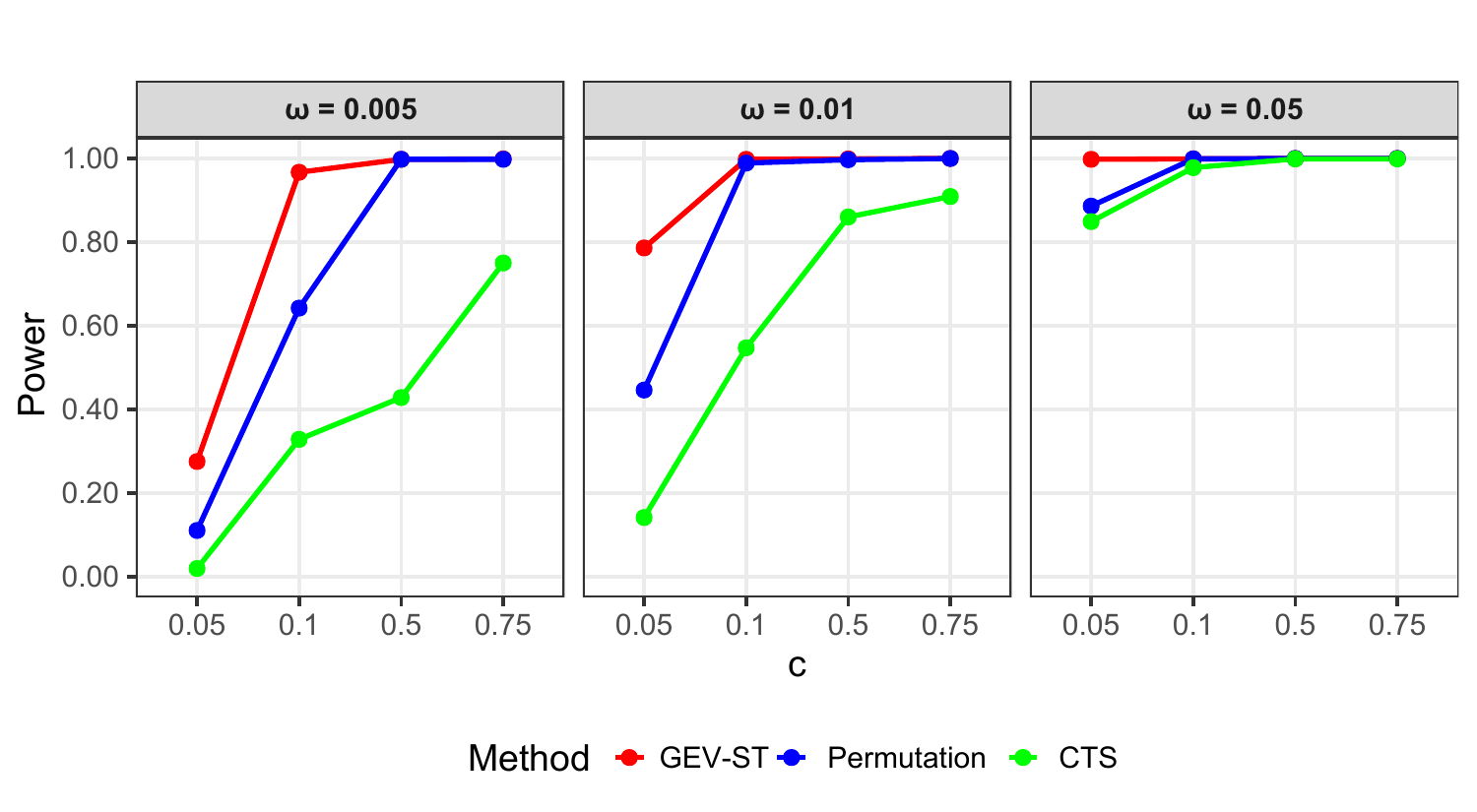}
		\caption{Empirical power under exponential baseline intensity with $a=0.05$, $b=0.5$ for different window sizes and signal strengths $c$.}
		\label{f-6}
	\end{minipage}
	\hfill
	\begin{minipage}[t]{0.48\linewidth}
		\centering
		\includegraphics[width=\linewidth]{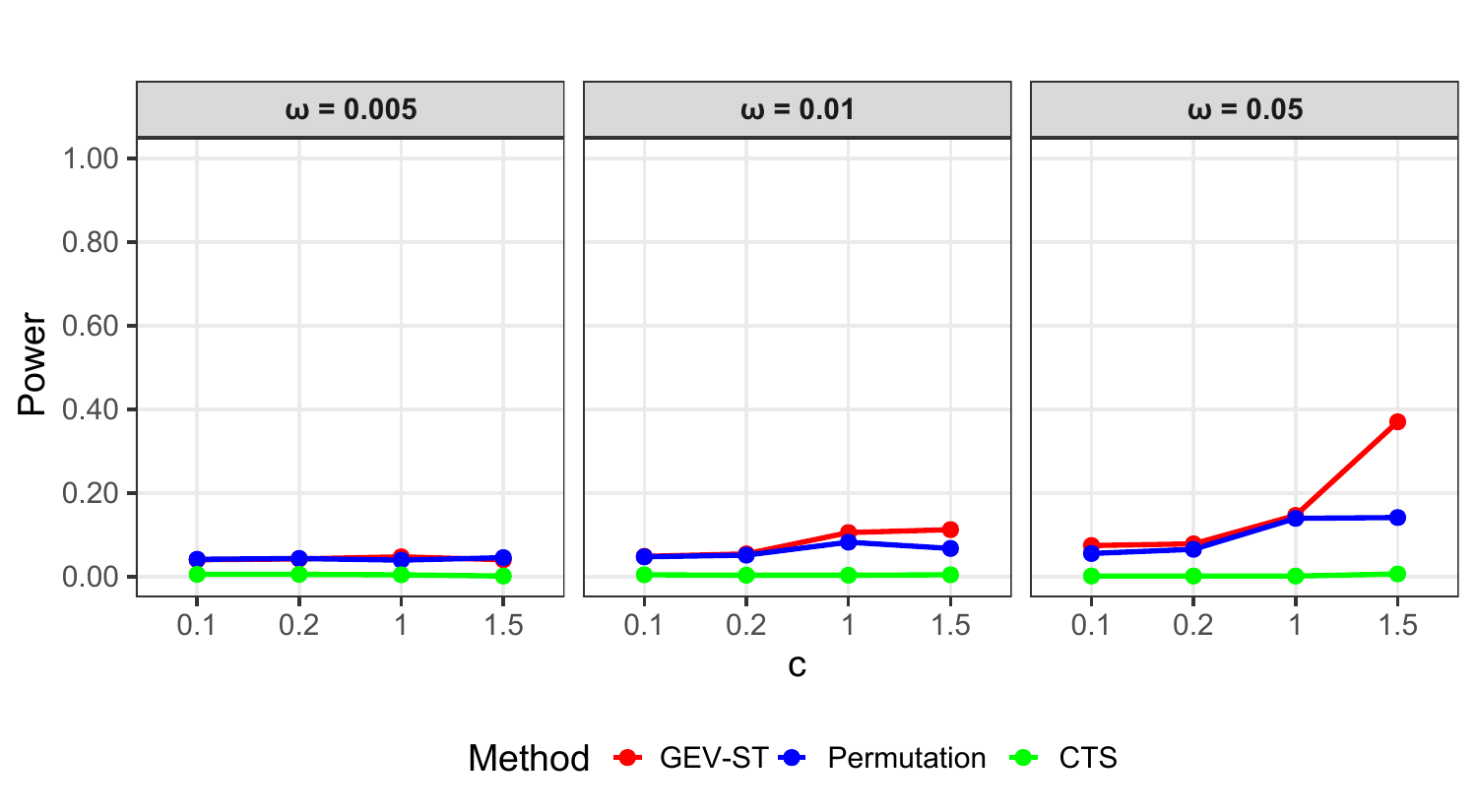}
		\caption{Empirical power under exponential baseline intensity with $a=0.5$, $b=0.05$ for different window sizes and signal strengths $c$.}
		\label{f-7}
	\end{minipage}
\end{figure}

\begin{figure}[htp!]
	\centering
	\begin{minipage}[t]{0.48\linewidth}
		\centering
		\includegraphics[width=\linewidth]{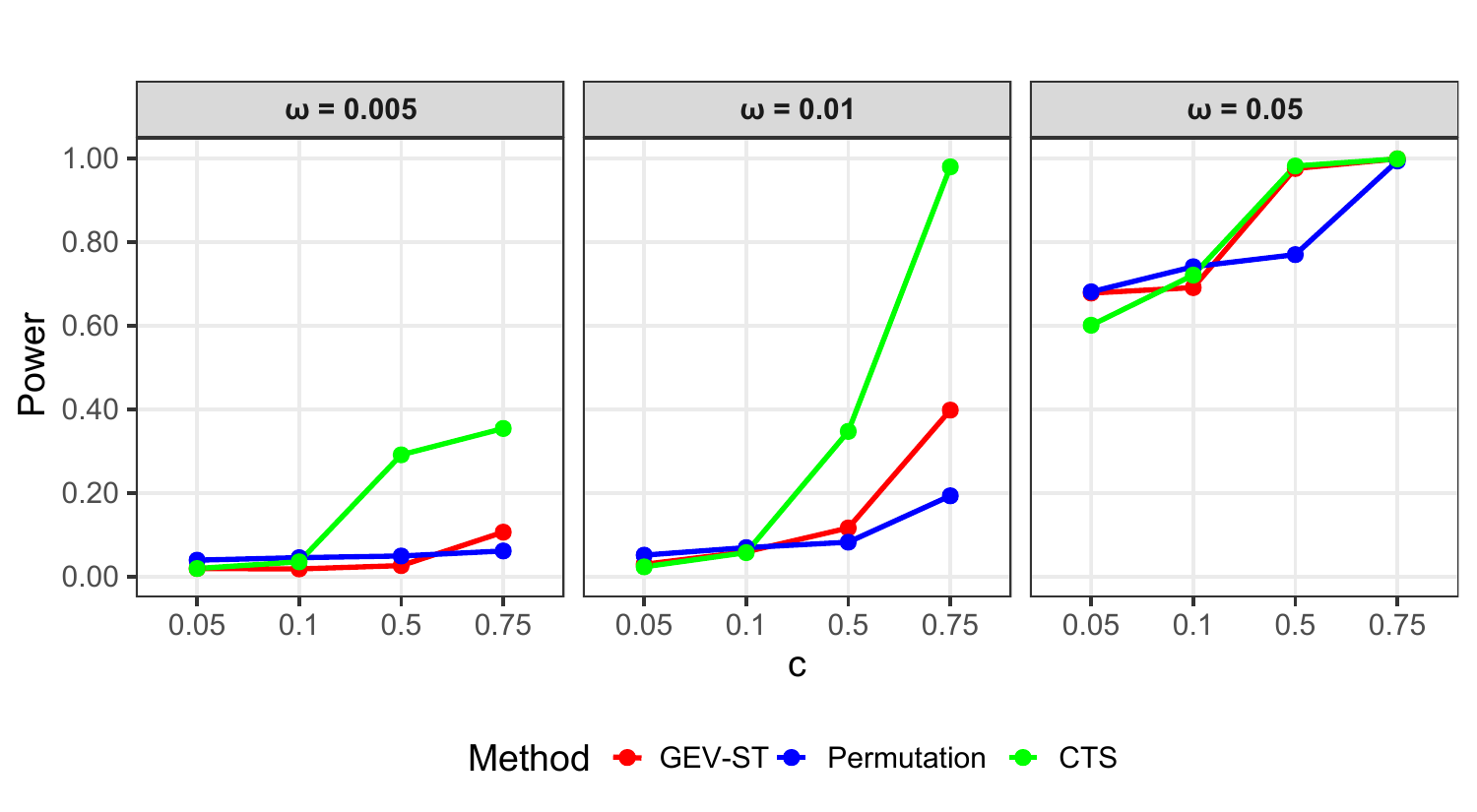}
		\caption{Empirical power under exponential baseline intensity with $a=1$, $b=2$ for different window sizes and signal strengths $c$.}
		\label{f-8}
	\end{minipage}
	\hfill
	\begin{minipage}[t]{0.48\linewidth}
		\centering
		\includegraphics[width=\linewidth]{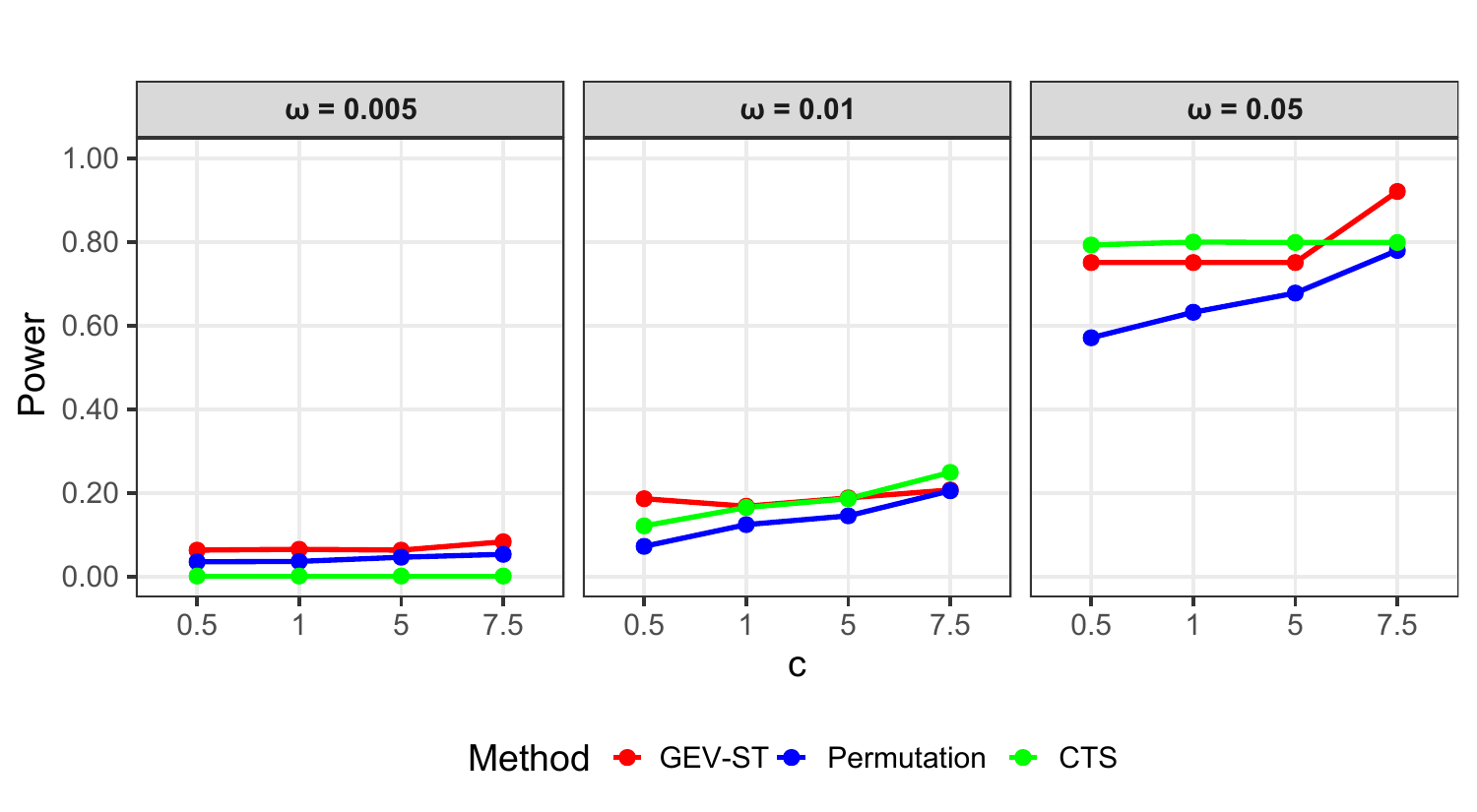}
		\caption{Empirical power under exponential baseline intensity with $a=1$, $b=5$ for different window sizes and signal strengths $c$.}
		\label{f-9}
	\end{minipage}
\end{figure}

\section{Application to Copy Number Variation Detection}

CNVs are genomic regions in which the number of DNA copies
differs from that of a reference genome. In tumor genomes, structural mutations such
as deletions and duplications produce localized changes in sequencing read intensity.
Next-generation sequencing (NGS) technology generates millions of short reads that are
mapped to genomic positions, and these mapped read locations can be naturally modeled
as realizations of point processes. In particular, CNVs correspond to localized changes
in the relative intensity between tumor and normal samples. This two-sample point-process
structure provides a natural framework for scan statistic-based cluster detection.

\subsection{Data description and model formulation}
We analyze the synthetic sequencing data distributed with the SeqCBS package, which mimic the HCC1954/BL1954 tumor-normal sequencing setting studied by \cite{ShenZhang2012}. Shen and Zhang developed a change-point model for nonhomogeneous Poisson processes and applied it to DNA copy number profiling. In the original HCC1954/BL1954 data set, after read alignment and quality filtering, approximately \(7.72\) million reads were obtained from the tumor sample and \(6.65\) million reads from the normal sample. These mapped read locations can be viewed as realizations of two independent nonhomogeneous Poisson processes.
 
The synthetic sequencing data used in this analysis are available in the SeqCBS
package in R. The package provides simulated tumor and normal read locations that
mimic real next-generation sequencing data with copy number changes embedded at
known genomic regions. For the sequence used in our analysis, the normal sample
contains \(n_0 = 15193\) mapped read locations, and the tumor sample contains
\(n_1 = 16452\) mapped read locations. For these sample sizes, we use the scanning
window lengths
\[
\omega \in \{0.005, 0.0075, 0.01\}
\]
on the transformed \([0,1]\) scale. These choices  are motivated by the window-size conditions in Theorem~\ref{thm:gev-grid-scan}, since the expected null tumor counts are moderately large and the corresponding quadratic terms remain relatively small $(n_1\omega = 82.26, 123.39, 164.52;\ n_1\omega^2 = 0.41, 0.93, 1.65)$. Thus, these window lengths provide a finite-sample implementation of the shrinking-window regime used in the extreme-value approximation.

We analyze these read locations using the proposed scan statistic methods, including the
GEV approximation and the permutation-based calibration.
The scan statistic is computed over moving windows along the ordered genomic positions.
Large scan statistic values indicate local excess tumor reads relative to the normal
baseline and are identified as candidate CNV regions. To identify the detected cluster regions, every scanning window whose count exceeds the critical value is retained. Any retained windows that overlap are merged into a single interval. The resulting nonoverlapping intervals are reported as the final detected clusters.

Figures~\ref{f-10}--\ref{f-12} show the detected CNV regions using the three window sizes $\omega=0.005$, $0.0075$, and $0.01$, respectively. Across these window sizes, the GEV-ST and permutation methods consistently identify the same three main CNV regions, and the detected intervals are close to those reported by SeqCBS. This indicates that the proposed scan-statistic methods are stable across these nearby choices of $\omega$. The empirical CDF plots show clear separation between the tumor and normal samples in the detected regions, confirming the presence of copy number changes. The CTS method also detects the main CNV regions; however, for $\omega=0.0075$ and $\omega=0.01$, it identifies additional regions beyond the three main ones.

Figures~\ref{f-13}--\ref{f-15} show the corresponding sliding-window  count plots for the three window sizes. The scan statistic curves from GEV-ST and permutation show clear peaks at the same genomic locations identified by SeqCBS. These peaks correspond to regions where tumor read counts are higher than the normal baseline. As the window size increases, the peaks become slightly broader, but the main detected locations remain stable. CTS also shows peaks at similar locations, although its peaks are sometimes broader or include additional neighboring regions. Overall, the GEV-ST and permutation methods produce results that are consistent with CTS and SeqCBS.

Tables~\ref{T-1}--\ref{T-3} report the detected CNV regions using the three window sizes. The columns give the start and end locations of each detected region for GEV-ST, Permutation, SeqCBS, and CTS. The proposed GEV-ST and permutation methods detect the same main CNV regions across all three window sizes, with only minor changes in the reported interval endpoints. In contrast, CTS detects additional regions when $\omega=0.0075$ and $\omega=0.01$, suggesting that it may be more sensitive to neighboring or broader local departures for these larger window sizes.

\begin{figure}[h!]
	\centering
	\includegraphics[width=.75\linewidth]{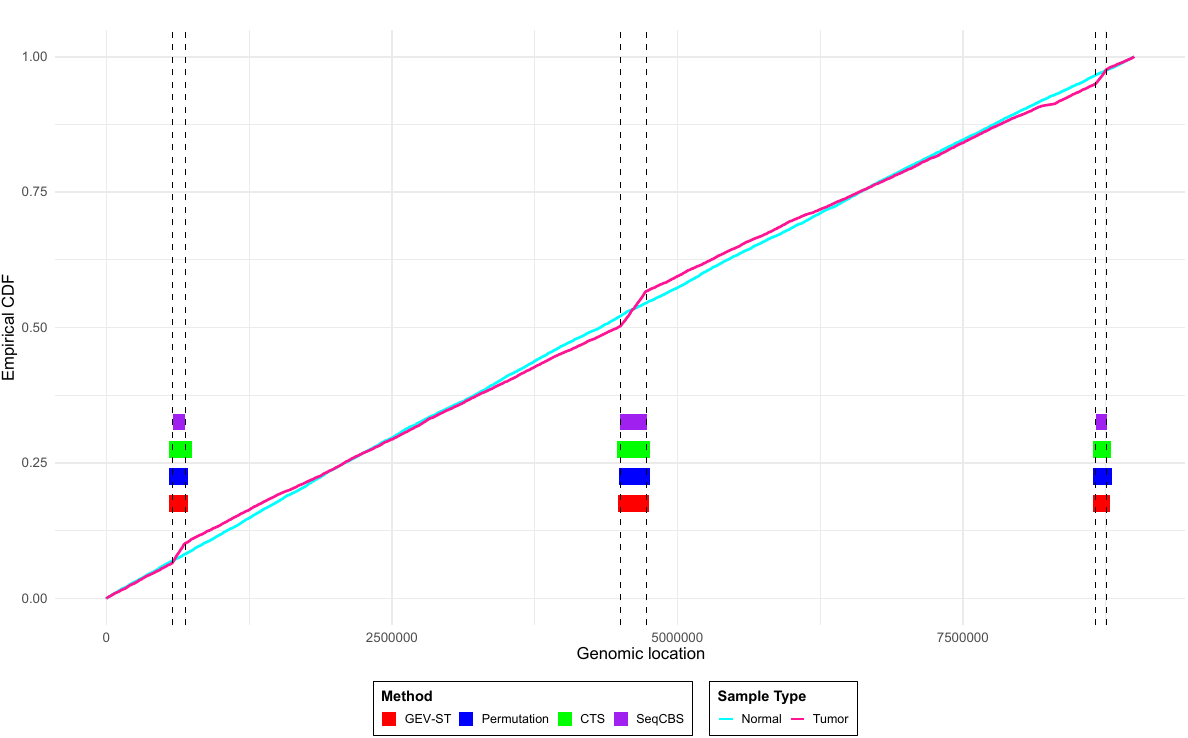}
	\caption{CNV detection with window size $\omega = 0.005$.}
	\label{f-10}
\end{figure}
\begin{figure}[h!]
	\centering
	\includegraphics[width=.75\linewidth]{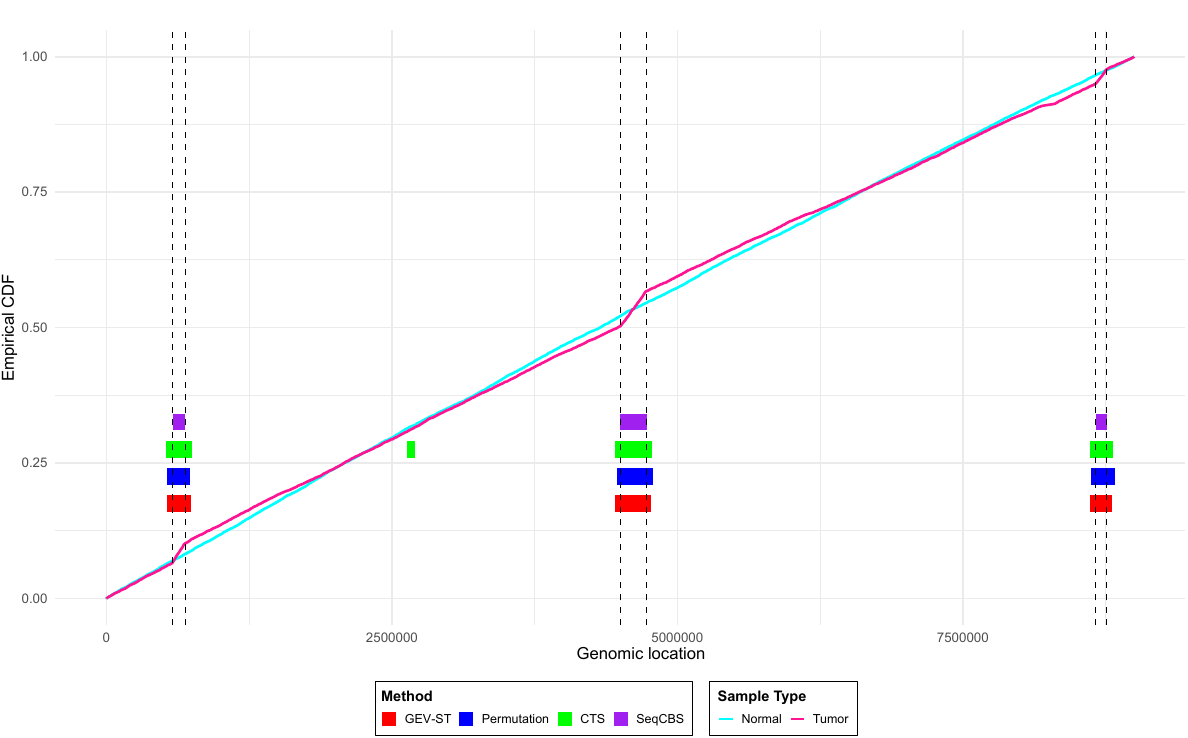}
	\caption{CNV detection with window size $\omega = 0.0075$.}
	\label{f-11}
\end{figure}
\begin{figure}[h!]
	\centering
	\includegraphics[width=.75\linewidth]{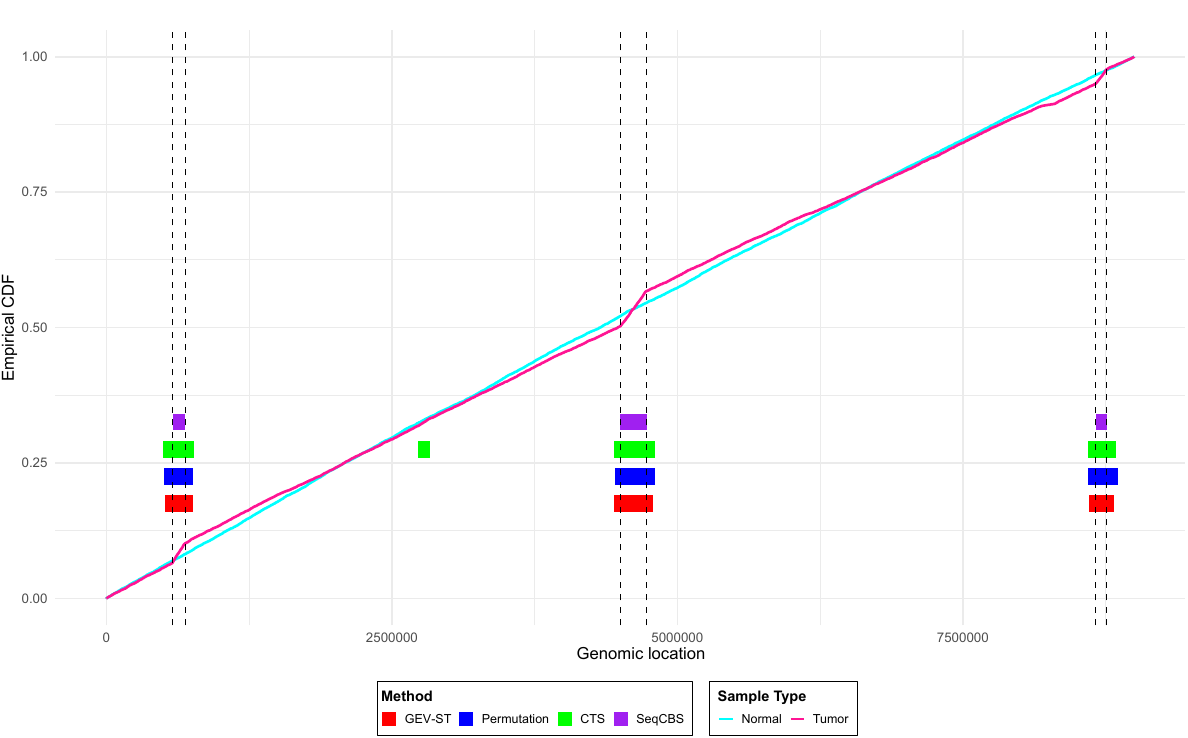}
	\caption{CNV detection with window size $\omega = 0.01$.}
	\label{f-12}
\end{figure}

\begin{figure}[h!]
	\centering
	\includegraphics[width=.75\linewidth]{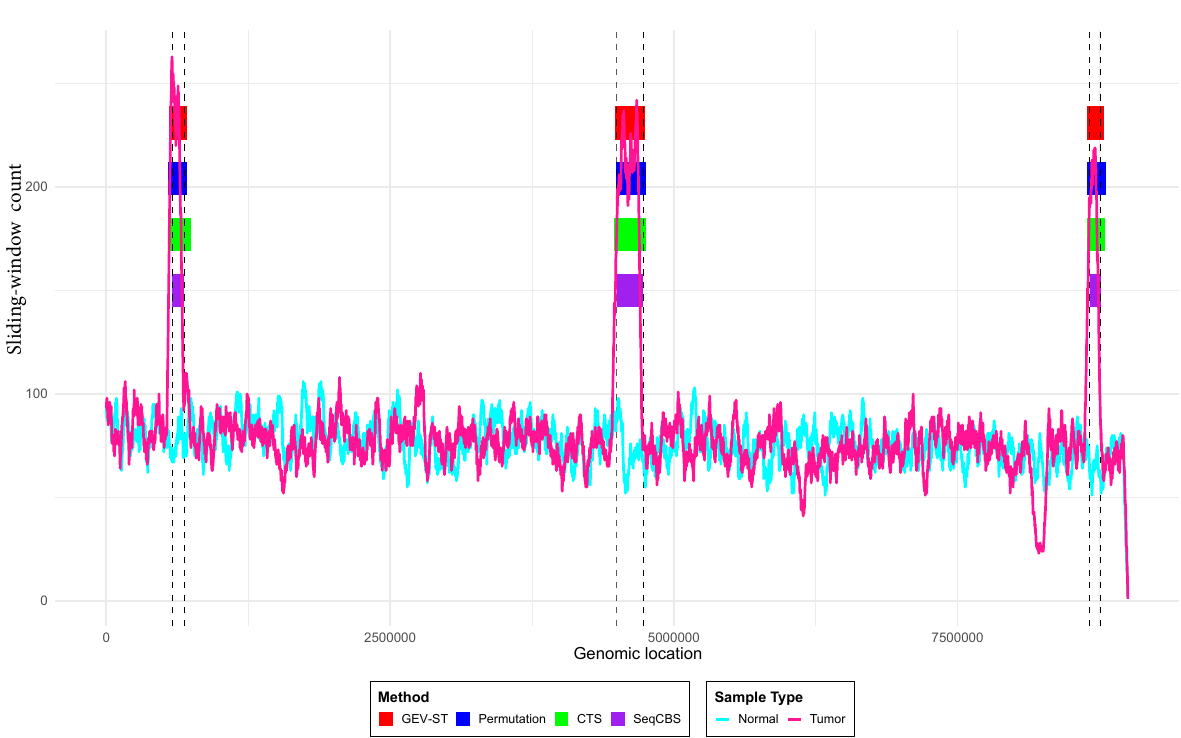}
	\caption{Sliding-window  count with window size  $\omega = 0.005$.}
	\label{f-13}
\end{figure}

\begin{figure}[h!]
	\centering
	\includegraphics[width=.75\linewidth]{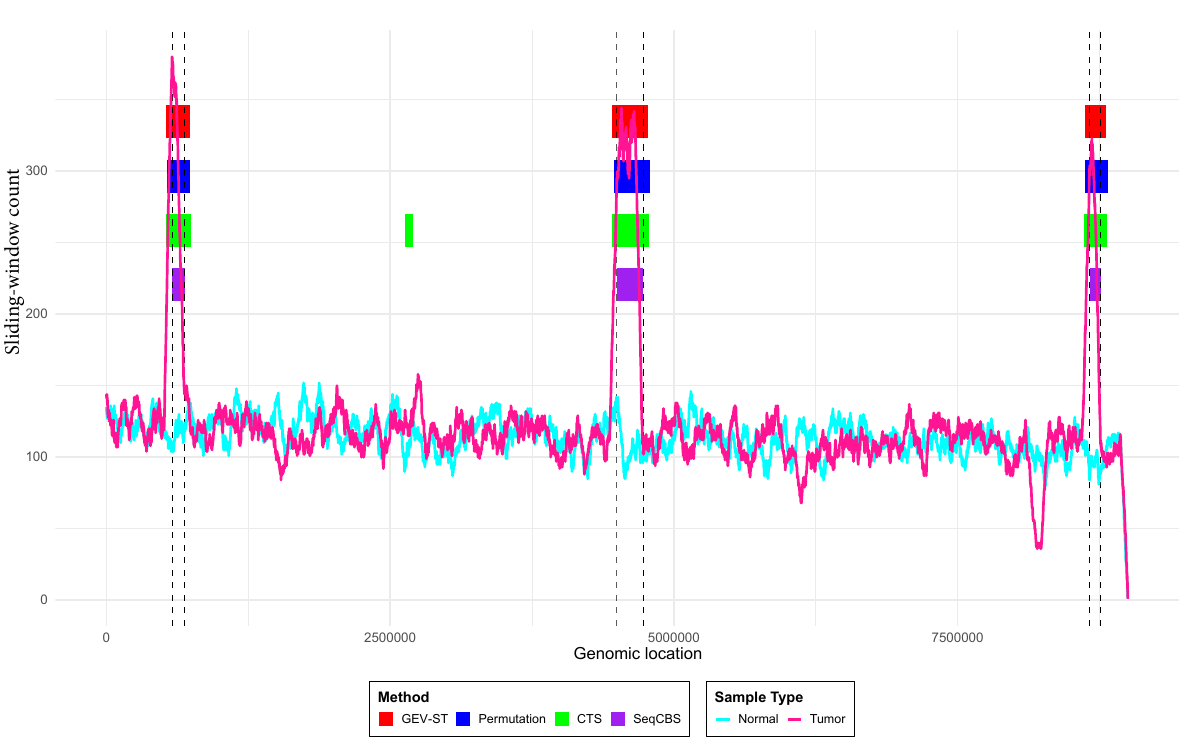}
	\caption{Sliding-window  count with window size  $\omega = 0.0075$.}
	\label{f-14}
\end{figure}

\begin{figure}[h!]
	\centering
	\includegraphics[width=.75\linewidth]{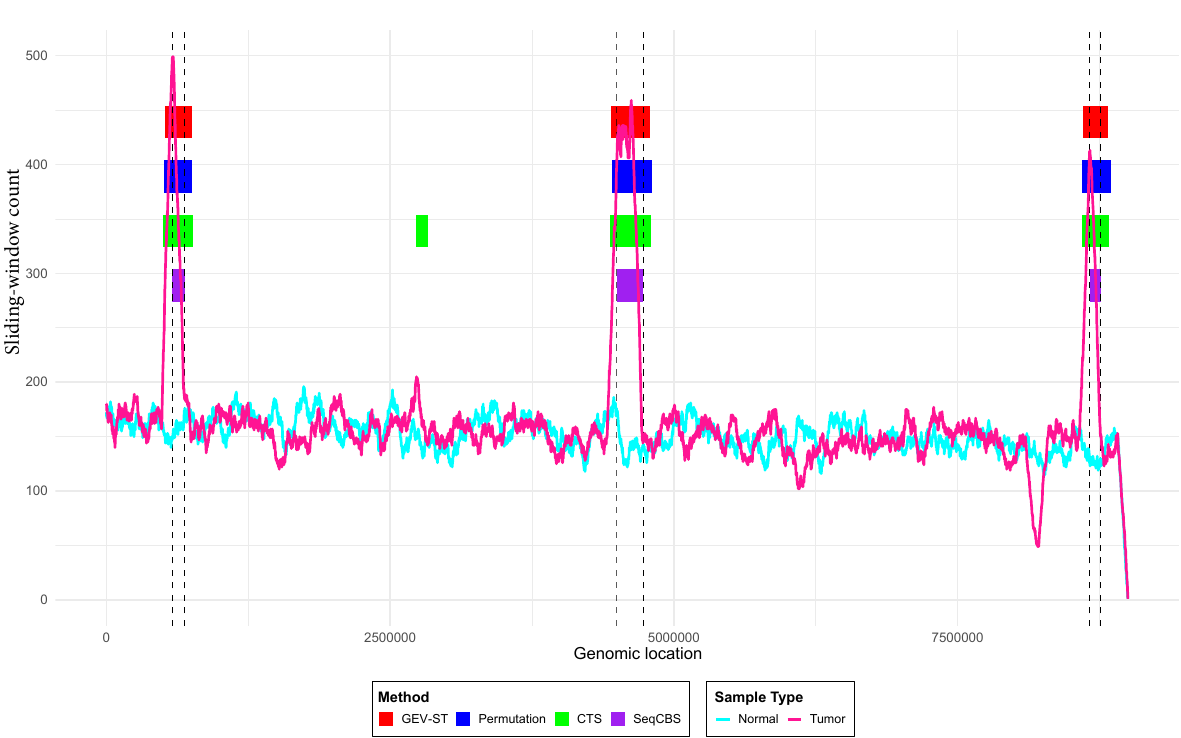}
	\caption{Sliding-window  count with window size  $\omega = 0.01$.}
	\label{f-15}
\end{figure}

\begin{table}[h!]
\centering
\caption{Detected CNV regions using $\omega=0.005$. Start and End denote the genomic start and end locations. SeqCBS does not use window scanning.}
\label{T-1}
\begin{tabular}{c|cc|cc|cc|cc}
\toprule
Region &
\multicolumn{2}{c|}{GEV-ST} &
\multicolumn{2}{c|}{Permutation} &
\multicolumn{2}{c|}{SeqCBS} &
\multicolumn{2}{c}{CTS} \\
\cmidrule(r){2-3}
\cmidrule(r){4-5}
\cmidrule(r){6-7}
\cmidrule(r){8-9}
&Start & End &
Start & End &
Start & End &
Start & End \\
\midrule
1&549825  & 710327  & 545780  & 711310  & 578871  & 688700  & 544798  & 749795 \\
2&4483825 & 4748198 & 4489440 & 4756271 & 4497332 & 4730251 & 4473509 & 4757102 \\
3&8640982 & 8788162 & 8641733 & 8810311 & 8662374 & 8759728 & 8636554 & 8799374 \\
\bottomrule
\end{tabular}
\end{table}

   \begin{table}[h!]
   	
   	\centering
   	
   	\scriptsize
   	
   	\caption{Detected CNV regions using $\omega=0.0075$. SeqCBS does not use window scanning.}
   	
   	\label{T-2}
   	
   	\begin{tabular}{c|cc|cc|cc|cc}	
   		\toprule
   		Region &
   		\multicolumn{2}{c|}{GEV-ST} &
   		\multicolumn{2}{c|}{Permutation} &
   		\multicolumn{2}{c|}{SeqCBS} &
   		\multicolumn{2}{c}{CTS} \\
   		\cmidrule(r){2-3}
   		\cmidrule(r){4-5}
   		\cmidrule(r){6-7}
   		\cmidrule(r){8-9}
   		& Start & End & Start & End & Start & End & Start & End \\
   		
   		\midrule
   		
   		1 & 527830  & 738214  & 531691  & 735953  & 578871  & 688700  & 521143  & 748603 \\
   		
   		2 & 4458055 & 4771001 & 4473509 & 4785853 & 4497332 & 4730251 & 2631487 & 2698773 \\
   		
   		3 & 8618019 & 8808950 & 8621253 & 8828725 & 8662374 & 8759728 & 4455868 & 4781786 \\
   		
   		4 & ---     & ---     & ---     & ---     & ---     & ---     & 8617420 & 8818110 \\
   		
   		\bottomrule
   		
   	\end{tabular}
   	
   \end{table}
  
 \begin{table}[h!]
\centering
\caption{Detected CNV regions using $\omega=0.01$. SeqCBS does not use window scanning.}
\label{T-3}
\begin{tabular}{c|cc|cc|cc|cc}
\toprule
Region &
\multicolumn{2}{c|}{GEV-ST} &
\multicolumn{2}{c|}{Permutation} &
\multicolumn{2}{c|}{SeqCBS} &
\multicolumn{2}{c}{CTS} \\
\cmidrule(r){2-3}
\cmidrule(r){4-5}
\cmidrule(r){6-7}
\cmidrule(r){8-9}
 & Start & End & Start & End & Start & End & Start & End \\
\midrule
1 & 512394  & 753829  & 505182  & 756358  & 578871  & 688700  & 498456  & 764640 \\
2 & 4448253 & 4787006 & 4452259 & 4807960 & 4497332 & 4730251 & 2725904 & 2831301 \\
3 & 8606486 & 8822612 & 8594559 & 8855197 & 8662374 & 8759728 & 4441101 & 4800680 \\
4 & ---     & ---     & ---     & ---     & ---     & ---     & 8592703 & 8837444 \\
\bottomrule
\end{tabular}
\end{table}

\FloatBarrier 

\section{Summary and Discussion}

This paper develops a scan statistic framework for cluster detection in a two-sample nonhomogeneous Poisson process (NHPP) setting, motivated by copy number variation (CNV) detection from next-generation sequencing data. The main methodological contribution is to combine a time transformation with an extreme-value approximation to obtain a usable null distribution for the scan statistic under heterogeneous baselines. This allows the scan statistic to be applied directly to point-process data without requiring binning or parametric modeling of the baseline intensity. To complement the asymptotic calibration, we also introduce a permutation procedure that provides a fully nonparametric significance assessment.

The theoretical results show that, under the null transformation, the scan statistic converges to a GEV distribution with an extremal-index correction that accounts for dependence among overlapping windows. This provides a practical way to estimate critical values using maximum likelihood estimation of the GEV parameters together with an estimate of the extremal index. The limiting result in Theorem 1 is established for the oracle transformation based on the true distribution function \(F_X\). In practice, \(F_X\) is replaced by the empirical distribution function \(\widehat F_X\), and the transformed case observations are \(V_j=\widehat F_X(Y_j)\). For the same limiting approximation to remain valid for the empirical-CDF-transformed grid scan statistic, the estimation error in \(\widehat F_X\) should be small relative to the scale of the scan statistic after normalization. In particular, the perturbation caused by replacing \(F_X\) with \(\widehat F_X\) should not change the maximum window count at the order of the normalizing scale used in the extreme-value limit. Under such a condition, the empirical transformation has no first-order effect on the limiting distribution. A full theoretical treatment of this empirical-transformation error is an important direction for future work.

Simulation results in Section~3 demonstrate that both GEV-ST and permutation methods provide reliable performance under a wide range of heterogeneous baseline intensities. The permutation method maintains Type I error rates close to the nominal level across the considered window sizes and intensity settings. The GEV approximation gives comparable calibration, but it shows slightly inflated Type I error when the window size is large. This behavior is consistent with the asymptotic theory, because larger windows increase overlap among neighboring scan windows and move the finite-sample setting farther from the shrinking-window regime assumed in Theorem~\ref{thm:gev-grid-scan}. Under alternatives with embedded clusters, GEV-ST and permutation generally show
strong power for small and moderate window sizes, especially for weaker signals
under slowly or moderately increasing baselines. For rapidly increasing exponential
baselines and larger windows, CTS can become more competitive and may achieve
higher power.

The CNV application in Section~4 further illustrates the practical performance of the proposed methods. For the analyzed sequence, the normal and tumor samples have sizes $n_0=15193$ and $n_1=16452$, respectively. The application uses the window sizes $\omega=0.005$, $0.0075$, and $0.01$ on the transformed $[0,1]$ scale, which are motivated by the window-size conditions in Theorem~\ref{thm:gev-grid-scan}. Across these window sizes, the GEV-ST and permutation methods consistently identify the same three main CNV regions, with detected intervals close to those reported by SeqCBS. The corresponding empirical CDF plots show clear separation between the tumor and normal samples in the detected regions, and the sliding-window  count plots show clear peaks at the same genomic locations. These results indicate that the proposed methods are stable across nearby choices of $\omega$ in the finite-sample setting.

The comparison with CTS also highlights a difference between the proposed scan-statistic methods and the continuous testing approach. CTS detects the main CNV regions, but for $\omega=0.0075$ and $\omega=0.01$, it also identifies additional regions beyond the three main ones. This suggests that CTS may be more sensitive to neighboring or broader local departures, whereas the GEV-ST and permutation methods give a more stable summary of the dominant CNV regions. Thus, the proposed methods provide results that are consistent with SeqCBS and CTS while maintaining a direct scan-statistic interpretation based on the maximum local tumor-read excess.

The simulation and application results suggest practical guidance for choosing between the two calibration methods. For small or moderate window sizes that are close to the shrinking-window regime, the GEV approximation is attractive because it is computationally efficient and can provide competitive power while maintaining reasonable Type I error control. The permutation method is useful as a robust nonparametric calibration, especially when finite-sample accuracy is a priority or when the dependence among overlapping windows is stronger. 


\section*{Statements and Declarations}

\subsection*{Funding}
The authors did not receive support from any organization for the submitted work.

\subsection*{Competing interests}
The authors have no relevant financial or non-financial interests to disclose.

\subsection*{Ethics approval}
Not applicable.

\subsection*{Consent to participate}
Not applicable.

\subsection*{Consent for publication}
Not applicable.

\subsection*{Data availability}
The synthetic sequencing data analyzed in this study are available from the SeqCBS package in R. The simulation settings and analysis procedures are described in the manuscript.

\subsection*{Code availability}
The code used to implement the proposed methods and reproduce the numerical results is available from the corresponding author upon reasonable request.

\subsection*{Author contributions}
Asanka R. Duwage and Tung-Lung Wu contributed equally to this work. Both authors read and approved the final manuscript.

\bibliographystyle{sn-basic}      
\bibliography{reference}

\end{document}